\documentclass[12pt]{article}

\usepackage[colorlinks, linkcolor=blue, citecolor=blue]{hyperref}           
\usepackage{color}
\usepackage{graphicx,subfigure,amsmath,amssymb,amsfonts,bm,epsfig,epsf,url,dsfont}
\usepackage{amsthm}

\newtheorem{thm}{Theorem}[section]
\newtheorem{lem}[thm]{Lemma}
\newtheorem{assum}[thm]{Assumption}
\newtheorem{proposition}[thm]{Proposition}
\newtheorem{remark}[thm]{Remark}
\newtheorem{definition}[thm]{Definition}

\newtheorem{corollary}[thm]{Corollary}

\newtheorem{problem}[thm]{Problem}

\usepackage{tikz}
\usepackage{bbm}
\usepackage[small,bf]{caption}
\usepackage{natbib}
\usepackage[top=1in,bottom=1in,left=1in,right=1in]{geometry}
\usepackage{fancybox}
\usepackage{hyperref}
\usepackage{algorithm}
\usepackage{verbatim}

\usepackage{threeparttable}
\usepackage[titletoc,toc]{appendix}

\usepackage{mathtools}

\usepackage{scalerel,stackengine}
\stackMath
\newcommand\reallywidehat[1]{%
\savestack{\tmpbox}{\stretchto{%
  \scaleto{%
    \scalerel*[\widthof{\ensuremath{#1}}]{\kern-.6pt\bigwedge\kern-.6pt}%
    {\rule[-\textheight/2]{1ex}{\textheight}}
  }{\textheight}%
}{0.5ex}}%
\stackon[1pt]{#1}{\tmpbox}%
}

\newcommand*{\rom}[1]{\expandafter\@slowromancap\romannumeral #1@}

\DeclareMathOperator*{\argmin}{arg min}

\usepackage{xspace}

\numberwithin{equation}{section}

\allowdisplaybreaks
\usepackage[utf8]{inputenc}
\usepackage{authblk}
\usepackage{url}
\usepackage[left=1in,right=1in,top=1in,bottom=1in]{geometry}

\usepackage{titlesec}
\titleformat{\part}[display]
  {\centering\bfseries\LARGE}
  {}{0pt}{}
\titlespacing*{\part}{0pt}{1.0em}{1.0em}

\usepackage{booktabs}
\usepackage{array}
\usepackage{enumitem}
\newcolumntype{L}[1]{>{\raggedright\arraybackslash}p{#1}}
\setlist[itemize]{leftmargin=*,nosep,topsep=0pt,parsep=0pt,partopsep=0pt}

\makeatletter
\renewcommand\paragraph{\@startsection{paragraph}{4}{\z@}%
  {1ex}
  {-1em}
  {\normalfont\normalsize\bfseries}}
\makeatother

\begin{document}
\title{Projected Multi-Reference Alignment}

\author[1]{Amnon Balanov\thanks{Corresponding author: \url{amnonba15@gmail.com}}}

\author[3]{Josh Katz}
\author[1]{Tamir Bendory}
\author[2]{Dan Edidin}

\affil[1]{\normalsize School of Electrical and Computer Engineering, Tel Aviv University, Tel Aviv 69978, Israel}

\affil[2]{Department of Mathematics, University of Missouri, Columbia, MO 65211, USA}

\affil[3]{The MITRE Corporation, Bedford, 01730, MA, USA}

\maketitle

\begin{abstract}

Motivated by structural biology applications, we study the projected multi-reference alignment (MRA) model introduced in~\cite{bandeira2023estimation}, in which an unknown signal is observed through noisy samples, each generated by applying a random cyclic shift followed by a fixed projection.
The projection merges reflection-symmetric index pairs, thereby discarding orientation information.
The goal is to recover the dihedral orbit of the signal.
We prove that in the high-noise regime, the first three  moments of the projected observations determine a generic dihedral orbit.
The main mechanism is a reduction, at the moment level, from projected MRA to the reflection-invariant phase-coupling structure of dihedral MRA. In Fourier-cosine coordinates adapted to the projection, the first moment determines the mean component, the second moment determines the Fourier magnitudes, and selected third moments yield the cosine phase-coupling relations appearing in the dihedral bispectrum. These relations lead to a constructive recovery scheme from moments up to order three.
We complement the population theory with finite-sample experiments comparing expectation--maximization (EM), direct moment optimization, and direct Fourier-cosine moment optimization. The results show that, in the high-noise regime, both EM and direct moment optimization are consistent with the predicted third-moment sample-complexity scaling $n \gtrsim \sigma^6$, where $n$ is the number of observations and $\sigma^2$ is the noise variance.
\end{abstract}

\renewcommand{\thefootnote}{$\ast$}
\footnotetext{Approved for Public Release; Distribution Unlimited. Public Release Case Number 26-1056. Affiliation with the MITRE Corporation is for identification purposes only and is not intended to convey or imply MITRE’s concurrence with, or support for, the positions, opinions, or viewpoints expressed by the
author. \copyright 2026 The MITRE Corporation. ALL RIGHTS RESERVED.}

\section{Introduction}

Multi-reference alignment (MRA) concerns the recovery of an unknown signal from noisy observations generated under random group actions~\cite{bandeira2023estimation,bendory2017bispectrum,abbe2018estimation,bandeira2014multireference,bandeira2020optimal,chen2018spectral,bendory2022dihedral,fan2023likelihood,ghosh2021multi,hirn2021wavelet,katsevich2023likelihood,romanov2021multi,perry2019sample}. In its classical cyclic form, one observes randomly shifted copies of an unknown one-dimensional signal and seeks to recover the signal up to the intrinsic cyclic-shift symmetry. Over the past decade, MRA has become a central model in statistical signal processing, both as a tractable orbit-recovery problem and as a simplified proxy for more complex inverse problems in imaging and structural biology.

In many imaging problems, however, the observations are not fully transformed copies of the unknown object. Rather, each measurement is a lower-dimensional projection of a transformed object. A prominent example is single-particle cryogenic electron microscopy (cryo-EM), where each image is modeled as a noisy tomographic projection of a three-dimensional volume at an unknown viewing direction~\cite{singer2020computational,bendory2020single}. Such projections can change the invariant structure of the recovery problem: by discarding directional information, they may introduce additional symmetries and enlarge the symmetry group of the recoverable object. A familiar example is the handedness, or chirality, ambiguity in cryo-EM: without additional information such as tilt-pair data, projection images do not determine the absolute hand of the reconstructed volume~\cite{rosenthal2003optimal}. Thus, projection does not merely reduce dimension; it can enlarge the inherent symmetry of the solutions so that the object is identifiable only up to a larger symmetry than the one present in the original transformation model.

\paragraph{The projected multi-reference alignment model.}
Compared with classical MRA, orbit recovery from projected observations is less well understood, and general identifiability results for projection models remain limited~\cite{balanov2026group}. As an initial finite-dimensional setting in which projection-induced effects can be analyzed explicitly, we study the \emph{projected} MRA model introduced in~\cite{bandeira2023estimation}; see Figure~\ref{fig:1}(a) for an illustration of the observation model.

Formally, let $p=2q+1$ be odd, and let $\theta\in\mathbb{R}^p$ be an unknown real signal indexed by the cyclic grid $\mathbb{Z}/p$. Each observation is obtained by applying a random cyclic shift to $\theta$, followed by a fixed projection that merges reflection-symmetric index pairs. Thus the observations $y_i\in\mathbb{R}^q$ are given by
\begin{align}
    y_i=\Pi(R_{\ell_i}\theta)+\xi_i,
    \label{eqn:main-projected-MRA}
\end{align}
where $R_{\ell_i}$ is the cyclic-shift operator, $\ell_i\sim \mathrm{Unif}(\mathbb{Z}/p)$, $\Pi:\mathbb{R}^p\to\mathbb{R}^q$ is the projection that sums each pair of indices $j$ and $-j$ on the cyclic grid; see~\eqref{eq:projection-map-model}; and $\xi_i$ is additive noise. 

Although the latent transformation in~\eqref{eqn:main-projected-MRA} is only a cyclic shift, the projection enlarges the effective symmetry of the problem. By merging reflection-symmetric index pairs, the projection masks the orientation of the cyclic ordering: traversing the underlying signal clockwise or counterclockwise gives the same projected orbit. Consequently, the signal is identifiable only up to the dihedral action generated by cyclic shifts and reflection, and the 
recovery goal is to infer the dihedral orbit of $\theta$ from the observations $y_1,\dots,y_n$;
see Figure~\ref{fig:1}(b) for an illustration. In this sense, projected MRA exhibits a basic feature of projection-based inverse problems: the group acting in the data-generating model may be smaller than the symmetry group of the identifiable solution.

At a structural level, this model shares with cryo-EM the form ``random transformation followed by projection and noise'': an unknown object is acted on by an unobserved group element, and only a lower-dimensional projection of the transformed object is measured. The analogy is not literal: here the projection is a reflection-symmetric pairing map on an odd one-dimensional cyclic grid, rather than a tomographic projection of a three-dimensional volume. However, the value of the model is that it isolates, in a tractable finite-dimensional setting, how projection can change both the identifiable symmetry and the low-order moment information available for recovery.

\begin{figure}[t!]
    \centering
    \includegraphics[width=0.9\linewidth]{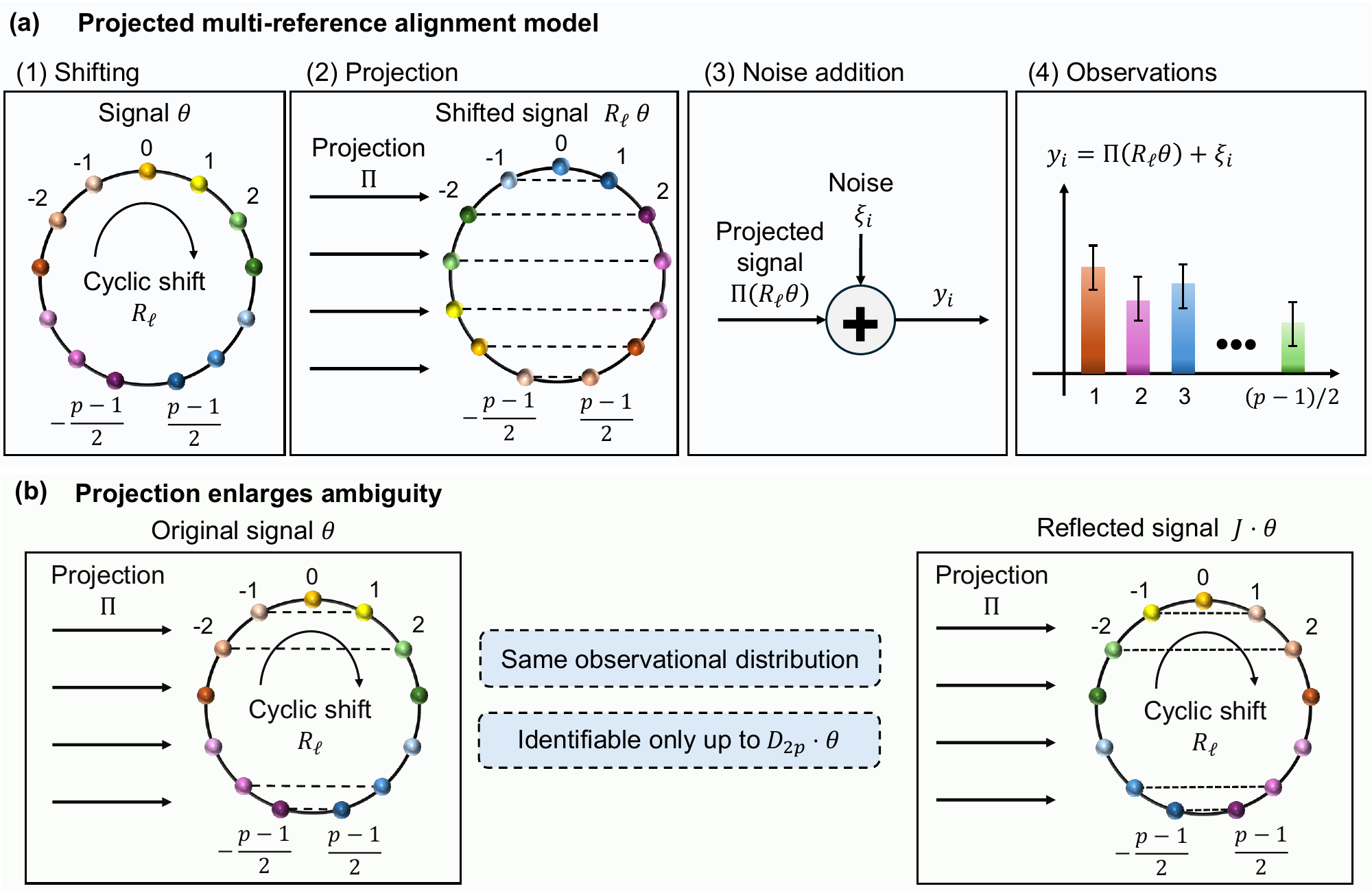}
    \caption{ \textbf{Projected multi-reference alignment (MRA) model and projection-induced ambiguity}. \textbf{(a)} Each observation is generated by applying an unknown cyclic shift $R_{\ell_i}$ to the signal $\theta$, projecting by $\Pi$~\eqref{eq:projection-map-model}, and adding noise: $y_i=\Pi(R_{\ell_i}\theta)+\xi_i$. For $p=2q+1$, the projection merges the reflection-symmetric index pairs $j$ and $-j$, yielding a $q$-dimensional projected observation. The vertical bars in the observation panel schematically represent the uncertainty induced by additive noise.
    \textbf{(b)} The projection masks orientation information: the original signal $\theta$ and its reflected version $J\theta$ generate the same projected observation distribution. Hence the signal is identifiable only up to the dihedral orbit $D_{2p}\cdot\theta$, rather than merely up to cyclic shift. }
    \label{fig:1}
\end{figure}

\paragraph{Moment-based recovery.}
In the high-noise regime, MRA-type problems are typically approached through low-order invariant moments of the observation distribution rather than by estimating the unknown group elements. In cyclic MRA, these invariants include the mean, power spectrum, and bispectrum~\cite{bandeira2023estimation,perry2019sample,bendory2017bispectrum,chen2018spectral}. The bispectrum carries Fourier phase information, and frequency-marching or bispectrum-inversion methods provide exact phase recovery in the noiseless population setting~\cite{bendory2017bispectrum,chen2018spectral}. More generally, the first moment order that determines the underlying orbit is closely tied to the sample complexity of MRA in the high-noise regime.

Within projected MRA, the central question is whether the low-order moment information that suffices for recovery in cyclic MRA is retained after projection. The associated third-moment identifiability problem was left open in~\cite{bandeira2023estimation}: prior work gave computational evidence and conjectured that moments up to degree three should generically restrict the solution set to finitely many candidates, but did not provide an identifiability proof or a constructive recovery algorithm. 

\paragraph{Connection with dihedral MRA.}
A central point of our analysis is the close relation between projected MRA and dihedral MRA~\cite{edidin2026reflection,bendory2022dihedral}. In dihedral MRA, observations are generated by applying random cyclic shifts and reflections to the unknown signal, so the recovery target is explicitly a dihedral orbit. Recent work on the reflection-invariant bispectrum~\cite{edidin2026reflection} shows that, for generic signals, degree-three dihedral invariants determine the dihedral orbit and analyzes the corresponding reflection-invariant phase equations.

The projected MRA model considered here is not the same statistical observation model. In dihedral MRA, one averages over automorphisms of the signal space, namely,  cyclic shifts and reflections. In projected MRA, by contrast, one averages over the rank-$q$ operators $\Pi R_\ell$, which first shift the signal and then map it to a lower-dimensional reflection-symmetric observation space. Nevertheless, we show that the low-order averages carry the same degree-three reflection-invariant information needed for dihedral recovery. 
Concretely, the projection combines the contribution of each Fourier mode with that of its reflected conjugate mode. As a result, the complex bispectral phase relations of cyclic MRA~\cite{bendory2017bispectrum} are replaced by cosine phase-coupling relations, exactly as in the reflection-invariant bispectrum for dihedral MRA. Thus, our main task is to prove that these degree-three dihedral phase relations are retained by the lower-dimensional projected observations.

\paragraph{Main results.}
The main purpose of the present paper is to prove that, for odd signal dimension $p=2q+1$, 
the first three population moments of the projected observations in~\eqref{eqn:main-projected-MRA} determine the underlying signal up to the intrinsic dihedral ambiguity, under a generic nondegeneracy condition. The proof is constructive and proceeds by identifying, within the projected model, the same reflection-invariant phase information that appears in dihedral MRA~\cite{edidin2026reflection}; see Figure~\ref{fig:2} for an illustration of the recovery pipeline.

Our approach has two steps. First, we rewrite the projected observations in a Fourier-cosine representation adapted to the reflection-symmetric projection; see Definition~\ref{def:projected-fourier-cosine-coefficients}. In this representation, the first moment determines the mean component, the second moment determines the Fourier magnitudes, and selected components of third moments yield explicit cosine phase-coupling relations among the Fourier phases. 

Second, we use these cosine relations to obtain an explicit recovery procedure. The chain equations propagate phase information from low to high frequencies, but each cosine equation introduces a binary sign ambiguity. Crucially, this ambiguity does not require an exhaustive search over all sign patterns. In contrast to the direct dihedral phase-recovery procedure in~\cite{edidin2026reflection}, where the sign ambiguities lead to an exponential branch search, our consistency relations allow the signs to be recovered sequentially, with linear rather than exponential complexity in the number of frequencies. After fixing a reflection normalization, the first consistency relation determines the initial sign triple, and each subsequent local consistency relation determines one additional sign. Thus, the branch-and-prune step terminates sequentially: the active normalized branch is pruned at each stage to a single admissible continuation, while the only remaining global ambiguity is the intrinsic reflection ambiguity of the dihedral recovery problem. Once the sign branch is identified, the remaining phase parameter is recovered analytically from a terminal chain relation. See Figure~\ref{fig:2}(b) and Algorithm~\ref{alg:projected-mra-constructive-recovery} for the sequential branch-and-prune procedure.

We summarize the main conclusion informally as follows. The formal theorem statement is Theorem~\ref{thm:projected-mra-anchored-recovery}.

\begin{thm}[Informal theorem]
For odd signal dimension $p=2q+1\ge 7$, and under a generic nondegeneracy condition, the first three population moments of the projected MRA model defined in~\eqref{eqn:main-projected-MRA} determine the underlying signal up to its inherent dihedral ambiguity. 
\end{thm}

As a direct consequence of Theorem~\ref{thm:projected-mra-anchored-recovery}, the projected MRA model falls under the general moment-based sample-complexity principle for orbit-recovery problems: in the high-noise regime, if the lowest-order moments that carry the relevant orbit information have order $d$, then the sample complexity scales as $\sigma^{2d}$~\cite{bandeira2023estimation}. In the projected MRA model studied here, the first moment determines only the mean component, the second moment determines the Fourier magnitudes, and the phase information needed for orbit recovery first appears at order three. Thus, the natural high-noise scaling is $n\gtrsim \sigma^6$; see Figure~\ref{fig:3}(b) for illustration of this scaling law.

\begin{corollary}[High-noise sample-complexity scaling]
\label{cor:projected-mra-sample-complexity}
Consider the projected MRA model~\eqref{eqn:main-projected-MRA}, with shifts drawn uniformly from $\mathbb{Z}/p$, and let $\theta$ satisfy the generic conditions of Theorem~\ref{thm:projected-mra-anchored-recovery}. Then consistent recovery by any estimator requires at least $n/\sigma^6\to\infty$ samples, while estimators based on moments up to
order three attain the same leading scaling.
\end{corollary}

\paragraph{Organization of the paper.} 
The paper is organized as follows. In Section~\ref{sec:prelimiaries}, we introduce the projected MRA model, the Fourier-cosine representation adapted to the projection, and the low-order population moments used throughout the analysis. In Section~\ref{sec:projected-mra-third-moment-recovery}, we prove the main recovery theorem: we derive the relevant moment identities, formulate the sign-only branch-and-prune argument, and show that the third-order consistency constraints eliminate all spurious branches under the stated generic nondegeneracy assumptions. In Section~\ref{sec:finite-sample-analysis}, we study the finite-sample setting by comparing EM with direct moment and Fourier-cosine moment optimization. Section~\ref{sec:discussion} concludes with a discussion of the results and possible extensions.

\begin{figure}[t!]
    \centering
    \includegraphics[width=1.0\linewidth]{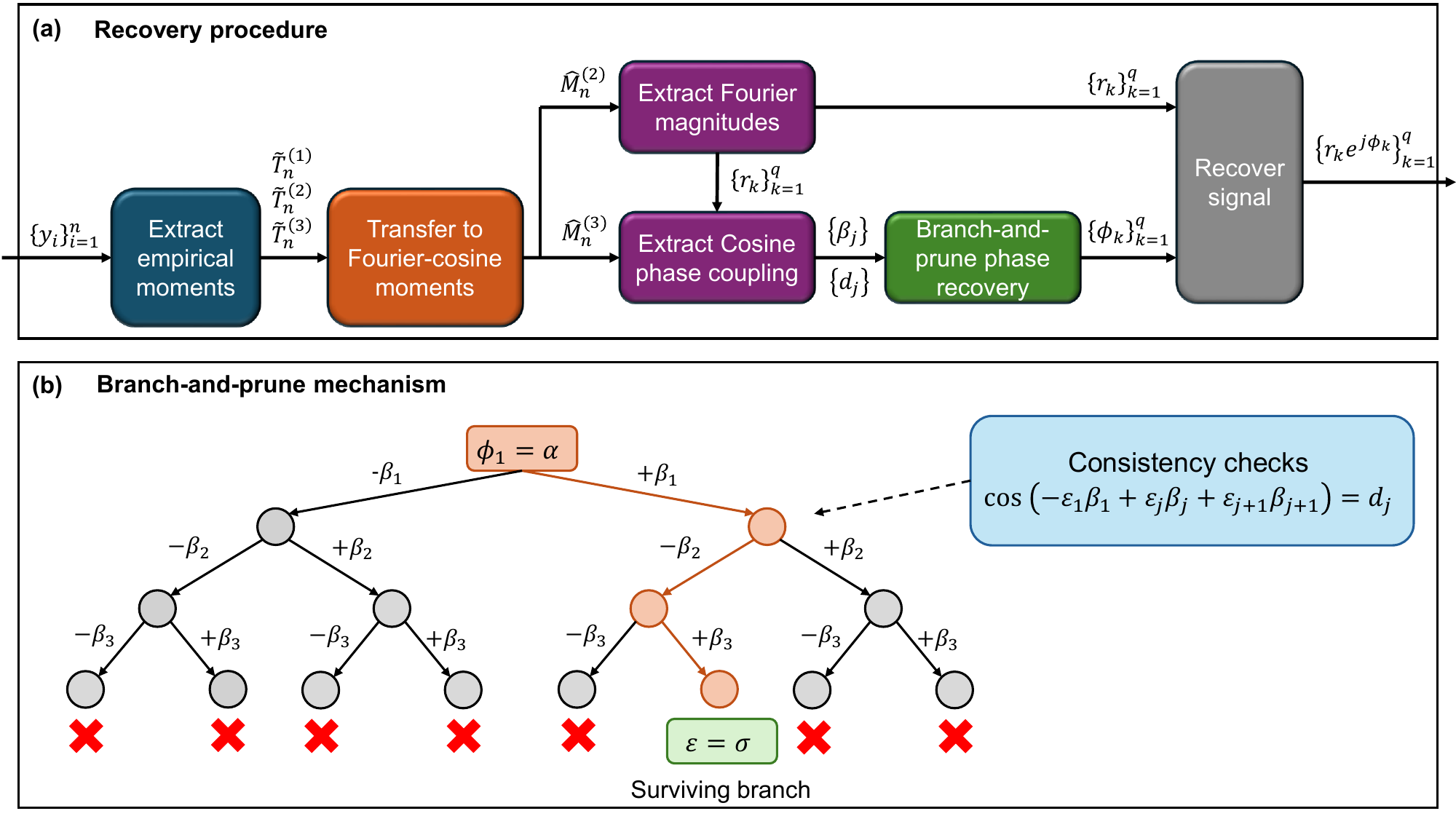}
    \caption{
    \textbf{Recovery pipeline and sequential branch-and-prune phase reconstruction for projected MRA.}
    \textbf{(a)} High-level recovery procedure. From the noisy observations $\{y_i\}_{i=1}^n$, we compute empirical moments and transfer them to Fourier-cosine moments by~\eqref{eq:transfer-projected-to-cosine-moments}. The second Fourier-cosine moment yields the Fourier magnitudes $\{r_k\}_{k=1}^q$ by Proposition~\ref{prop:second-moment-amplitudes}, while selected third Fourier-cosine moments yield the chain cosines $\{c_j\}_{j=1}^q$ in~\eqref{eq:cj}-\eqref{eq:cq} and the consistency cosines $\{d_j\},d_\ast$ in~\eqref{eq:dj}-\eqref{eq:dstar}. These quantities are used in the constructive recovery procedure of Algorithm~\ref{alg:projected-mra-constructive-recovery}.
    \textbf{(b)} Sequential branch-and-prune mechanism. The chain equations propagate the phases from an anchor $\phi_1=\alpha$, but each cosine relation introduces a binary sign ambiguity. The consistency relations then prune incompatible sign choices. After fixing the reflection normalization, e.g., by setting $\varepsilon_1=+1$, the first consistency relation determines the initial sign triple, and each subsequent local consistency relation determines one additional sign. Thus, in the population setting, the sign ambiguity is resolved sequentially rather than by an exhaustive search over $2^q$ branches; see Algorithm~\ref{alg:projected-mra-constructive-recovery}.
    }
\label{fig:2}
\end{figure}

\section{Preliminaries}
\label{sec:prelimiaries}
In this section, we introduce the projected MRA model, and define the low-order population moments that underlie the recovery result. We also derive a convenient Fourier representation of the projected orbit, which will allow us to pass from tensor moments in the observation space to explicit scalar relations in frequency space.

\subsection{The projected multi-reference alignment model}
We next formalize the model introduced in~\eqref{eqn:main-projected-MRA}.

\begin{problem}[Reflection-projected multi-reference alignment]
\label{prob:projected-mra}
Let $p\ge 3$ be odd, and set $q:=(p-1)/2$. We identify $\mathbb{R}^p$ with real-valued functions on $\mathbb{Z}/p$, indexed by the 
representative set $\{0,\pm1,\ldots,\pm q\}$. For each $\ell\in\mathbb{Z}/p$, let $R_\ell:\mathbb{R}^p\to\mathbb{R}^p$ denote the cyclic-shift operator
\begin{align}
    (R_\ell \theta)[m] = \theta[m-\ell],
    \qquad m\in\mathbb{Z}/p,
    \label{eq:cyclic-shift-operator}
\end{align}
where indices are understood modulo $p$. Define the reflection-symmetric
projection $\Pi:\mathbb{R}^p\to\mathbb{R}^q$ by
\begin{align}
    (\Pi v)[j] = v[j]+v[-j],
    \qquad j=1,\ldots,q.
    \label{eq:projection-map-model}
\end{align}
Given an unknown signal $\theta\in\mathbb{R}^p$, we observe independent samples $y_1,\ldots,y_n\in\mathbb{R}^q$ of the form
\begin{align}
    y_i=\Pi(R_{\ell_i}\theta)+\xi_i,
    \qquad i=1,\ldots,n,
    \label{eq:projected-mra-sample}
\end{align}
where $\ell_i\sim\mathrm{Unif}(\mathbb{Z}/p)$ are independent and
$\xi_i\in\mathbb{R}^q$ are additive noise variables. 
The \emph{projected MRA problem} is to infer $\theta$ from $y_1,\ldots,y_n$ up to the intrinsic dihedral symmetry generated by cyclic shifts and reflection, namely, to recover the orbit $D_{2p}\cdot\theta$.
\end{problem}

We formulate the model for odd signal length $p=2q+1$, as in the original projected MRA formulation of~\cite{bandeira2023estimation}. This assumption avoids the degeneracy that appears in the even-dimensional analog. When $p$ is odd, the nonzero indices of $\mathbb{Z}/p$ split into distinct reflection pairs $\{j,-j\}$, $j=1,\ldots,q$. By contrast, when $p=2q$ is even, $\theta$ and $\theta+(c,-c,c,-c,\ldots,c,-c)$ induce the same observation law for every scalar $c$. Hence, the even-dimensional model is not well posed as a problem of recovering the full dihedral orbit of $\theta$. In the frequency domain, this obstruction corresponds to the Nyquist frequency $q=p/2$, which satisfies $q\equiv -q \pmod p$ and is therefore fixed by reflection. Equivalently, in the time domain, this frequency is represented by the alternating mode $(1,-1,1,-1,\ldots)$, which is invisible to the projected shifted observations.
In Section~\ref{subsec:continuous-so2-extension}, we show that the same projection-induced dihedral ambiguity arises in the continuous $\mathrm{SO}(2)$ formulation. In the continuous setting, the nonzero Fourier modes occur in distinct pairs $k$ and $-k$, and there is no Nyquist frequency. Thus, the odd-dimensional discrete model is the natural finite analogue of the continuous projected $\mathrm{SO}(2)$ setting.
\subsection{Projected cosine representation}
For each $\ell\in\mathbb{Z}/p$, let
$X_\ell := \Pi(R_\ell\theta)\in\mathbb{R}^q$ denote the noiseless projected
observation. By~\eqref{eq:projection-map-model},
\begin{align}
    X_\ell[j] = \theta[j-\ell]+\theta[-j-\ell],
    \qquad j=1,\dots,q,
    \label{eq:projected-sample-entry}
\end{align}
with all indices understood modulo $p$.

Throughout, we use the unitary discrete Fourier transform (DFT) convention
\begin{align}
    \hat{x}[k] = \frac{1}{\sqrt{p}} \sum_{n=0}^{p-1} x[n]e^{-2\pi i kn/p},
    \qquad
    x[n] = \frac{1}{\sqrt{p}} \sum_{k=0}^{p-1} \hat{x}[k]e^{2\pi i kn/p}.
    \label{eq:unitary-dft-idft}
\end{align}
For real-valued $x$, $\hat{x}[-k]=\overline{\hat{x}[k]}$. Substituting the
inverse DFT into~\eqref{eq:projected-sample-entry} gives
\begin{align}
    X_\ell[j]
    &=
    \frac{1}{\sqrt{p}} \sum_{k=0}^{p-1} \hat{\theta}[k] \Bigl(e^{2\pi i k(j-\ell)/p} + e^{-2\pi i k(j+\ell)/p}\Bigr) \notag\\
    &=
    \frac{2}{\sqrt{p}} \sum_{k=0}^{p-1} \hat{\theta}[k] e^{-2\pi i k\ell/p} \cos\Bigl(\frac{2\pi kj}{p}\Bigr) \notag\\
    &=
    \frac{2}{\sqrt{p}}\,\hat{\theta}[0] + \frac{2}{\sqrt{p}}\sum_{k=1}^{q} \Bigl( \hat{\theta}[k]e^{-2\pi i k\ell/p} + \hat{\theta}[-k]e^{2\pi i k\ell/p} \Bigr) \cos\Bigl(\frac{2\pi kj}{p}\Bigr) \notag\\
    &=
    \frac{2}{\sqrt{p}}\,\hat{\theta}[0] + \frac{4}{\sqrt{p}}\sum_{k=1}^{q} \Re\!\Bigl\{ \hat{\theta}[k]e^{-2\pi i k\ell/p} \Bigr\} \cos\Bigl(\frac{2\pi kj}{p}\Bigr),
    \qquad j=1,\dots,q.
    \label{eq:Xl-entry-fourier}
\end{align}

\begin{definition}[Projected Fourier-cosine coefficients]
\label{def:projected-fourier-cosine-coefficients}
Let $p=2q+1$ be odd, and let $\theta\in\mathbb{R}^p$. For each shift $\ell\in\mathbb{Z}/p$ and each non-zero frequency $k=1,\dots,q$, define the \emph{projected Fourier-cosine coefficient}
\begin{align}
    C_\ell[k] &:= \hat{\theta}[k]e^{-2\pi i k\ell/p} + \hat{\theta}[-k]e^{2\pi i k\ell/p}
    \label{eq:Cl-definition}\\
    &= 2\,\Re\!\Bigl\{\hat{\theta}[k]e^{-2\pi i k\ell/p}\Bigr\}.
    \label{eq:Cl-real}
\end{align}
\end{definition}
With this definition
we 
re-write~\eqref{eq:Xl-entry-fourier} in matrix form. 

\begin{align}
    X_\ell = \frac{2}{\sqrt{p}}\,\hat{\theta}[0]\mathbf{1} + \frac{1}{\sqrt{p}}\,A\,C_\ell,
   \label{eq:Xl-reduced-matrix-form}
\end{align}
where $A \in \mathbb{R}^{q \times q}$ is the cosine matrix
\begin{equation} 
A = (A)_{jk} := 2\cos\Bigl(\frac{2\pi jk}{p}\Bigr) \label{eq:cosine-matrix}
\end{equation}
and $\mathbf{1}\in\mathbb{R}^q$ denotes the all-ones vector.

Equation~\eqref{eq:Xl-reduced-matrix-form} shows that, for each $\ell$, the projected sample $X_\ell\in\mathbb{R}^q$ is obtained from the projected Fourier-cosine coefficient vector $C_\ell\in\mathbb{R}^q$ by an affine change of coordinates. In particular, once $\hat{\theta}[0]$ is known, the vectors $X_\ell$ and $C_\ell$ carry equivalent information provided that $A$ is invertible, as proved next.

\begin{lem}[Invertibility and conditioning of the cosine matrix]
\label{lem:cosine-matrix-invertible}
Let $p=2q+1$ be odd, and let $A\in\mathbb{R}^{q\times q}$ be defined by~\eqref{eq:cosine-matrix}. Then:

\begin{enumerate}
    \item The matrix $A$ satisfies
    \begin{align}
        A^\top A = p I_q - 2\mathbf{1}\mathbf{1}^\top ,    \label{eq:cosine-matrix-gram}
    \end{align}
    and in particular $A$ is invertible. 
    
    \item For $q\ge 2$,  $\kappa_{2}(A) = \sqrt{p}$ where $\kappa_2(A)$ denotes the ratio of the largest to smallest singular value.
\end{enumerate}
   
\end{lem}

\begin{proof}[Proof of Lemma~\ref{lem:cosine-matrix-invertible}]
For $1\le k,m\le q$, using $2\cos a\cos b=\cos(a-b)+\cos(a+b)$, we have
\begin{align}
    (A^\top A)_{km}
    &=
    4\sum_{j=1}^q \cos\!\left(\frac{2\pi jk}{p}\right) \cos\!\left(\frac{2\pi jm}{p}\right) \notag\\
    &=
    2\sum_{j=1}^q \cos\!\left(\frac{2\pi j(k-m)}{p}\right) + 2\sum_{j=1}^q \cos\!\left(\frac{2\pi j(k+m)}{p}\right).
    \label{eq:ATA-cosine-sums}
\end{align}
Since $p=2q+1$, the roots-of-unity identity gives, for every
$r\not\equiv 0 \pmod p$,
\begin{align}
    1+2\sum_{j=1}^q
    \cos\!\left(\frac{2\pi jr}{p}\right) = 0,
    \qquad\text{and hence}\qquad
    \sum_{j=1}^q \cos\!\left(\frac{2\pi jr}{p}\right) = -\frac{1}{2}.
    \label{eq:nonzero-cosine-sum}
\end{align}

\emph{Diagonal terms.}
If $k=m$, then the first sum in~\eqref{eq:ATA-cosine-sums} equals $q$, while $k+m=2k\not\equiv0\pmod p$, so the second sum equals $-1/2$ by~\eqref{eq:nonzero-cosine-sum}. Therefore
\begin{align}
    (A^\top A)_{kk} = 2q-1 = p-2.
    \label{eq:ATA-diagonal-entry}
\end{align}

\emph{Off-diagonal terms.}
If $k\ne m$, then both $k-m$ and $k+m$ are nonzero modulo $p$, so both sums in~\eqref{eq:ATA-cosine-sums} equal $-1/2$ by~\eqref{eq:nonzero-cosine-sum}. Hence
\begin{align}
    (A^\top A)_{km} = -2,
    \qquad k\ne m.
    \label{eq:ATA-off-diagonal-entry}
\end{align}
Combining~\eqref{eq:ATA-diagonal-entry} and~\eqref{eq:ATA-off-diagonal-entry}, we obtain
\begin{align}
    A^\top A = pI_q-2\mathbf{1}\mathbf{1}^\top .
    \label{eq:ATA-gram-final}
\end{align}

The matrix $\mathbf{1}\mathbf{1}^\top$ has rank one and the non-zero
eigenvalue is $q$ since this is the trace of the matrix.
Hence $A^\top A = pI_q - 2\mathbf{1} \mathbf{1}^\top$
has eigenvalues $p-2q=1$ and $p$. Hence the singular values of $A$ are $1$ and $\sqrt p$, and therefore $\kappa_2(A)=\sqrt p$.
\end{proof}

Thus, the change of variables from projected samples to Fourier-cosine coefficients is invertible, with spectral condition number growing as $\sqrt p$. Although this conditioning is mild as a function of the signal dimension, it can have a visible finite-sample effect when empirical moments are transformed to Fourier-cosine coordinates. In Section~\ref{sec:finite-sample-analysis}, we quantify this effect numerically: applying $A^{-1}$ in each tensor mode amplifies sampling fluctuations, which affects the finite-sample behavior of algorithms that operate directly in the Fourier-cosine moment coordinates.

\subsection{Population moments}

The recovery result is formulated in terms of low-order moments of the noiseless projected orbit
$X_\ell=\Pi(R_\ell\theta)$.
We distinguish between two equivalent representations. The first consists of the \emph{moments of the projected MRA model} 
defined directly from $X_\ell\in\mathbb{R}^q$. The second consists of the \emph{Fourier-cosine population moments}, defined from the projected Fourier-cosine coefficient vector $C_\ell\in\mathbb{R}^q$. These two representations are related by the affine change of coordinates~\eqref{eq:Xl-reduced-matrix-form}. Since the zero Fourier mode is determined by the first projected moment and the matrix $A$ is invertible by Lemma~\ref{lem:cosine-matrix-invertible}, the two moment representations contain the same information. The Fourier-cosine representation is the one used for the population level recovery: its second moment gives the Fourier magnitudes, and selected entries of the third moments give the cosine phase-coupling relations.

\begin{definition}[Moments of the projected population]
\label{def:population-moments}
Let $\theta\in \mathbb{R}^p$, and let $\ell\sim \mathrm{Unif}(\mathbb{Z}/p)$. Recall the noiseless projected sample $X_\ell$ from~\eqref{eq:projected-sample-entry}. For each integer $d\ge 1$, the $d$-th \emph{projected population moment} is
defined as follows:
\begin{align}
    T^{(d)}(\theta) := \mathbb{E}\bigl(\Pi(R_\ell\theta)\bigr)^{\otimes{d}}  =\frac{1}{p}\sum_{\ell=0}^{p-1} X_\ell^{\otimes{d}}.
    \label{eq:population-moment-definition}
\end{align}
\end{definition}
The tensor $T^{(d)}(\theta)$ decomposes as a sum of pure tensors
\begin{align}
T^{(d)}(\theta) = \sum_{(k_1, \ldots, k_d)} T^{(d)}_{k_1, \ldots k_d}(\theta) e_{k_1} \otimes \ldots \otimes e_{k_d},
\end{align}
where 
\begin{equation} \label{eq:tdk}
T^{(d)}_{k_1, \ldots k_d}(\theta) = {1\over p}\sum_{\ell=0}^{p-1}\prod_{t=1}^dX_\ell[k_t].
\end{equation}
\begin{definition}[Fourier-cosine moments]
\label{def:cosine-moments}
Let $C_\ell=(C_\ell[1],\dots,C_\ell[q])^\top\in\mathbb{R}^q$ denote the projected Fourier-cosine coefficient vector defined in~\eqref{eq:Cl-definition}-\eqref{eq:Cl-real}. For $d\ge 1$ and $k_1,\dots,k_d\in\{1,\dots,q\}$, define the corresponding $d$-th \emph{Fourier-cosine population moment} by
\begin{align}
    M^{(d)}(\theta) :=  {1\over p} \sum_{\ell=0}^{p-1} C_\ell^{\otimes d}
    \label{eq:cosine-population-moment-definition}
\end{align}
\end{definition}
Once again, tensor $M^{(d)}(\theta)$ decomposes as a sum of pure tensors
\begin{align}
M^{(d)}(\theta) = \sum_{(k_1, \ldots, k_d)} M^{(d)}_{k_1, \ldots k_d}(\theta) e_{k_1} \otimes \ldots \otimes e_{k_d},
\end{align}
where 
\begin{equation} \label{eq:mdk}
M^{(d)}_{k_1, \ldots k_d}(\theta) = {1\over p}\sum_{\ell=0}^{p-1}\prod_{t=1}^d C_\ell[k_t].
\end{equation}

\paragraph{Transfer to the Fourier-cosine representation.}
The affine relation~\eqref{eq:Xl-reduced-matrix-form} gives
\begin{align}
    C_\ell = \sqrt{p}\,A^{-1} \left( X_\ell-\frac{2}{\sqrt p}\hat\theta[0]\mathbf{1} \right).
    \label{eq:C-from-X-affine}
\end{align}
Moreover, $\frac{1}{p}\sum_{\ell=0}^{p-1} X_\ell = \frac{2}{\sqrt p}\hat\theta[0]\mathbf{1}$, so $\hat\theta[0]$ is determined by $T^{(1)}(\theta)$. Thus, it follows from Lemma~\ref{lem:cosine-matrix-invertible} that the Fourier-cosine moments are obtained from the moments of the projected MRA observation by centering and applying the invertible linear map $\sqrt p\,A^{-1}$.

Concretely, define the centered moments
\begin{align}
    \bar{T}^{(d)}(\theta) := \mathbb{E}\big[(X_\ell-\mathbb{E}[X_\ell])^{\otimes{d}}\big].
    \label{eq:centered-projected-moments}
\end{align}
The tensor $\bar{T}^{(d)}(\theta)$ is determined by the  moments $T^{(1)}(\theta),\ldots,T^{(d)}(\theta)$. Applying $\sqrt p\,A^{-1}$ along each tensor mode gives
\begin{align}
    M^{(d)}_{k_1,\ldots,k_d}(\theta) = p^{d/2} \sum_{j_1,\ldots,j_d=1}^q (A^{-1})_{k_1 j_1}\cdots (A^{-1})_{k_d j_d} \bar{T}^{(d)}_{j_1,\ldots,j_d}(\theta).
    \label{eq:transfer-projected-to-cosine-moments}
\end{align}
Equivalently, $M^{(d)}(\theta) = p^{d/2}(A^{-1})^{\otimes d}\bar{T}^{(d)}(\theta)$, where $A^{-1}$ is applied once to each tensor mode. Thus, projected population moments up to order $d$ and Fourier-cosine population moments up to order $d$ are equivalent, up to centering by the first moment and the invertible change of coordinates induced by $A$.

\section{Main results}
\label{sec:projected-mra-third-moment-recovery}

In this section we prove that the first three population moments of the projected MRA model determine the underlying signal, under explicit generic nondegeneracy conditions, up to its inherent dihedral symmetry. 

\begin{thm}[Orbit recovery from the first three moments]
\label{thm:projected-mra-anchored-recovery}
Let $p=2q+1\ge 7$ be odd, and let $\theta\in\mathbb{R}^p$. Assume that $\widehat{\theta}[k]\neq 0$ for every $k=1,\ldots,q$, and that Assumption~\ref{ass:generic-sign-separation} holds. Then, the population moments $T^{(1)}(\theta), T^{(2)}(\theta), T^{(3)}(\theta)$ determine the dihedral orbit $D_{2p}\cdot\theta$, namely, if $\theta'\in\mathbb{R}^p$ satisfies $T^{(d)}(\theta')=T^{(d)}(\theta)$ for $d=1,2,3$, then $\theta'\in D_{2p}\cdot\theta$.
\end{thm}

We prove the theorem in Section~\ref{sec:proof_thm_projected-mra-anchored-recovery}. The argument proceeds in three steps. First, we extract from the first three moments the mean component, the nonzero Fourier magnitudes, and a family of explicit third-order cosine phase-coupling relations. These relations are the projected-MRA analogue of the degree-three dihedral bispectral equations~\cite{edidin2026reflection}. Second, we organize a selected subset of these relations into a frequency-marching scheme, which propagates the Fourier phases from low to high frequencies up to sign choices. Finally, we use additional third-order consistency constraints to eliminate all incompatible sign branches, leaving only the unavoidable global reflection ambiguity.

\subsection{From  moments to dihedral phase relations}

Throughout, we work in the projected Fourier-cosine representation introduced in~\eqref{eq:cosine-population-moment-definition}. 
For $k=1,\dots,q$, write the nonzero Fourier coefficients of the unknown signal $\theta$ in polar form as
\begin{align}
    \hat{\theta}[k] = r_k e^{i\phi_k},
    \qquad k=1,\dots,q,
    \label{eq:polar-form}
\end{align}
where $r_k:=|\hat{\theta}[k]|$ is the magnitude of the $k$-th Fourier coefficient and $\phi_k\in\mathbb{R}$ is its phase. In the recovery theorem below, we assume $r_k>0$ for $k=1,\dots,q$, so that the normalized phase relations are well defined. 

\paragraph{First moment.}
We begin with the first moment. By the definition of the projected population moment in~\eqref{eq:population-moment-definition} and the Fourier representation of $X_\ell$ in~\eqref{eq:Xl-entry-fourier}, averaging over $\ell\in\mathbb{Z}/p$ eliminates all nonzero Fourier modes. Hence, 
\begin{align}
    T^{(1)}(\theta) = \frac{1}{p}\sum_{\ell=0}^{p-1} X_\ell = \frac{2}{\sqrt p}\,\hat{\theta}[0]\mathbf{1}.
    \label{eq:M1-fourier}
\end{align}
Thus, the first population moment determines the mean component $\hat{\theta}[0]$.

\paragraph{Second moment.}
We next show that the second projected Fourier-cosine moment recovers the power spectrum. This is the degree-two information that is invariant under both cyclic shifts and reflection, and hence is common to the cyclic and dihedral MRA viewpoints.

\begin{proposition}[Second moment]
\label{prop:second-moment-amplitudes}
For $1\le k,m\le q$,
\begin{align}
    M^{(2)}_{k,m}(\theta) = 2\,|\hat{\theta}[k]|^2\,\delta_{km}.
    \label{eq:second-moment-formula}
\end{align}
\end{proposition}

\begin{proof}[Proof of Proposition~\ref{prop:second-moment-amplitudes}]
By definition of the projected Fourier-cosine coefficients~\eqref{eq:Cl-definition}, $C_\ell[k] = \hat{\theta}[k]e^{-2\pi i k\ell/p} + \hat{\theta}[-k]e^{2\pi i k\ell/p}$.
Using the reality condition $\hat{\theta}[-k]=\overline{\hat{\theta}[k]}$, we obtain
\begin{align}
    C_\ell[k]\,C_\ell[m]
    &=
    \hat{\theta}[k]\hat{\theta}[m]e^{-2\pi i (k+m)\ell/p} + \hat{\theta}[k]\overline{\hat{\theta}[m]}e^{-2\pi i (k-m)\ell/p}
    \notag\\
    &\qquad
    + \overline{\hat{\theta}[k]}\hat{\theta}[m]e^{2\pi i (k-m)\ell/p} + \overline{\hat{\theta}[k]}\,\overline{\hat{\theta}[m]}e^{2\pi i (k+m)\ell/p}.
    \label{eq:second-moment-product-expansion}
\end{align}
Averaging~\eqref{eq:second-moment-product-expansion} over $\ell\in\mathbb{Z}/p$, only terms with zero frequency modulo $p$ remain. Since $1\le k,m\le q$ and $p=2q+1$, the congruences $k+m\equiv0\pmod p$ and $-(k+m)\equiv0\pmod p$ are impossible, while $k-m\equiv0\pmod p$ is equivalent to $k=m$. Therefore,
\begin{align}
    M^{(2)}_{k,m}(\theta)
    &=
    \hat{\theta}[k]\overline{\hat{\theta}[m]}\,\delta_{km}
    +
    \overline{\hat{\theta}[k]}\hat{\theta}[m]\,\delta_{km}
    \notag\\
    &=
    2\,|\hat{\theta}[k]|^2\,\delta_{km},
    \label{eq:second-moment-proof-final}
\end{align}
which proves the claim.
\end{proof}

\paragraph{Third moment.}
We next turn to the third moment. The relevant entries are those for which the total Fourier frequency appearing after expansion is zero modulo $p$. For $1\le a,b,c\le q$, this occurs, up to permutation of the indices, in two cases: the non-wrapping relation $a+b=c\le q$, and the wrapping relation $a+b+c=p$. These entries yield the reflection-invariant, or dihedral, bispectral phase relations. In the first case the phase coupling is $\phi_a+\phi_b-\phi_c$, while in the second it is $\phi_a+\phi_b+\phi_c$.

\begin{proposition}[Third moment]
\label{prop:third-moment-cosines}
Recall the polar representation~\eqref{eq:polar-form}. Let $1\le a,b,c\le q$.
\begin{enumerate}
    \item If $a+b=c\le q$, then
    \begin{align}
        M^{(3)}_{a,b,c}(\theta) = 2r_a r_b r_c \cos(\phi_a+\phi_b-\phi_c). \label{eq:bispectrum-nonwrap}
    \end{align}

    \item If $a+b+c=p$, then
    \begin{align}
        M^{(3)}_{a,b,c}(\theta) = 2r_a r_b r_c \cos(\phi_a+\phi_b+\phi_c). \label{eq:bispectrum-wrap}
    \end{align}
\end{enumerate}
Moreover, up to permutation of $(a,b,c)$, all other third-moment entries vanish identically.
\end{proposition}

\begin{proof}[Proof of Proposition~\ref{prop:third-moment-cosines}]
By~\eqref{eq:mdk},
\begin{align}
    M^{(3)}_{a,b,c}(\theta) = \frac{1}{p}\sum_{\ell=0}^{p-1} C_\ell[a]\,C_\ell[b]\,C_\ell[c].
    \label{eq:third-moment-definition-proof}
\end{align}
Expanding the product in~\eqref{eq:third-moment-definition-proof} gives eight terms. Each term has an exponential factor whose frequency is one of $\pm a\pm b\pm c$.
Averaging over $\ell\in\mathbb{Z}/p$ eliminates all terms except those whose total frequency is zero modulo $p$.

Assume first that $a+b=c\le q$. Since $1\le a,b,c\le q$ and $p=2q+1$, the only zero-frequency terms are those corresponding to $a+b-c=0$ and $-a-b+c=0$.
Hence the surviving contribution is
\begin{align}
    M^{(3)}_{a,b,c}(\theta)
    &=
    \hat{\theta}[a]\hat{\theta}[b]\overline{\hat{\theta}[c]}
    +
    \overline{\hat{\theta}[a]\hat{\theta}[b]\overline{\hat{\theta}[c]}}
    \notag\\
    &= 2\Re\!\bigl(\hat{\theta}[a]\hat{\theta}[b]\overline{\hat{\theta}[c]}\bigr).
    \label{eq:third-moment-nonwrap-real}
\end{align}
Using~\eqref{eq:polar-form}, $\hat{\theta}[a]\hat{\theta}[b]\overline{\hat{\theta}[c]} = r_a r_b r_c e^{i(\phi_a+\phi_b-\phi_c)}$, and 
substituting into~\eqref{eq:third-moment-nonwrap-real} proves~\eqref{eq:bispectrum-nonwrap}.

Assume next that $a+b+c=p$. Then the only zero-frequency terms modulo $p$ are those corresponding to $a+b+c\equiv 0 \pmod p$ and $-(a+b+c)\equiv 0 \pmod p$.
Therefore,
\begin{align}
    M^{(3)}_{a,b,c}(\theta)
    &=
    \hat{\theta}[a]\hat{\theta}[b]\hat{\theta}[c] + \overline{\hat{\theta}[a]\hat{\theta}[b]\hat{\theta}[c]}
    \notag\\
    &=    2\Re\!\bigl(\hat{\theta}[a]\hat{\theta}[b]\hat{\theta}[c]\bigr).
    \label{eq:third-moment-wrap-real}
\end{align}
Using again~\eqref{eq:polar-form}, $\hat{\theta}[a]\hat{\theta}[b]\hat{\theta}[c] = r_a r_b r_c e^{i(\phi_a+\phi_b+\phi_c)}$, and substituting into~\eqref{eq:third-moment-wrap-real} proves~\eqref{eq:bispectrum-wrap}.

Finally, if a third-moment entry is nonzero, then at least one of the frequencies $\pm a\pm b\pm c$ must vanish modulo $p$. Since $1\le a,b,c\le q$ and $p=2q+1$, this can occur only when one index is the sum of the other two, or when $a+b+c=p$. Thus, up to permutation, the only nonzero entries are precisely the two cases above.
\end{proof}

\subsection{Chain equations and consistency constraints}

We now specialize the third-order cosine relations from Proposition~\ref{prop:third-moment-cosines} to a set of equations that can be used constructively. These are the same type of reflection-invariant phase relations that arise in dihedral MRA~\cite{edidin2026reflection}: the cosine equations determine phase couplings only up to sign, and the remaining third-order relations are used to resolve these sign ambiguities.

\paragraph{Chain equations.}
The first family consists of the \emph{chain equations}, obtained from the entries
\[
    M^{(3)}_{1,j,j+1}(\theta),\quad 1\le j\le q-1,
    \qquad
    M^{(3)}_{1,q,q}(\theta).
\]
Assume $r_k\neq 0$ for all $k=1,\dots,q$, so that the normalized third-moment quantities below are well defined. By Propositions~\ref{prop:second-moment-amplitudes} and~\ref{prop:third-moment-cosines}, the moments determine
\begin{align}
    c_j &:= \frac{M^{(3)}_{1,j,j+1}(\theta)}{2r_1r_jr_{j+1}} = \cos(\phi_1+\phi_j-\phi_{j+1}),
    \qquad 1\le j\le q-1,
    \label{eq:cj}
    \\
    c_q &:= \frac{M^{(3)}_{1,q,q}(\theta)}{2r_1r_q^2} = \cos(\phi_1+2\phi_q),
    \label{eq:cq}
\end{align}
where~\eqref{eq:cq} uses $1+q+q=p$. Define
\begin{align}
    \beta_j:=\arccos(c_j)\in[0,\pi],
    \qquad j=1,\dots,q.
    \label{eq:beta-definition}
\end{align}
Then, the chain equations take the form
\begin{align}
    \phi_1+\phi_j-\phi_{j+1}
    &\equiv \pm \beta_j \pmod{2\pi},
    \qquad 1\le j\le q-1,
    \label{eq:chain-equation-j}
    \\
    \phi_1+2\phi_q
    &\equiv \pm \beta_q \pmod{2\pi}.
    \label{eq:chain-equation-q}
\end{align}
Each equation contains a sign ambiguity because the cosine determines its argument only up to sign. 
We thus define a \emph{sign branch} to be a choice
\begin{align}
    \varepsilon=(\varepsilon_1,\dots,\varepsilon_q)\in\{\pm1\}^q
\end{align}
of signs in~\eqref{eq:chain-equation-j}-\eqref{eq:chain-equation-q}.

The next lemma records the phase propagation recursion. For a fixed sign branch $\varepsilon$ and an anchor value $\phi_1=\alpha$, the chain equations determine all remaining phases.

\begin{lem}[Anchored phase propagation]
\label{lem:anchored-frequency-marching}
Fix $\alpha\in\mathbb{R}$ and set $\phi_1=\alpha$. Let $\varepsilon=(\varepsilon_1,\dots,\varepsilon_q)\in\{\pm1\}^q$ be a sign branch for the chain equations~\eqref{eq:chain-equation-j}-\eqref{eq:chain-equation-q}. Then, the corresponding phase trajectory satisfies
\begin{align}
    \phi_2 &= 2\alpha-\varepsilon_1\beta_1,
    \label{eq:anchored-phi2}
    \\
    \phi_{j+1} &= \alpha+\phi_j-\varepsilon_j\beta_j,
    \qquad 2\le j\le q-1.
    \label{eq:anchored-recursion}
\end{align}
Equivalently,
\begin{align}
    \phi_k = k\alpha-\sum_{t=1}^{k-1}\varepsilon_t\beta_t,
    \qquad k=1,\dots,q.
    \label{eq:phi-explicit-anchored}
\end{align}
\end{lem}

\begin{proof}[Proof of Lemma~\ref{lem:anchored-frequency-marching}]
Substituting $\phi_1=\alpha$ into~\eqref{eq:chain-equation-j} gives \eqref{eq:anchored-phi2}. Similarly,~\eqref{eq:chain-equation-j} gives \eqref{eq:anchored-recursion}. Iterating this recursion yields \eqref{eq:phi-explicit-anchored}.
\end{proof}


\paragraph{Consistency constraints.}
For a fixed sign branch and a fixed anchor value $\phi_1$, the chain equations in Lemma~\ref{lem:anchored-frequency-marching} determine all remaining phases recursively. Thus, before imposing the additional consistency constraints, the chain equations generate $2^q$ candidate sign branches. The role of the consistency constraints below is to identify the admissible normalized branch.

For the true phase vector, let $\sigma=(\sigma_1,\dots,\sigma_q)\in\{\pm1\}^q$ denote the corresponding sign branch, so that
\begin{align}
    \phi_1+\phi_j-\phi_{j+1}
    &\equiv \sigma_j\beta_j \pmod{2\pi},
    \qquad 1\le j\le q-1,
    \label{eq:true-sign-j}
    \\
    \phi_1+2\phi_q
    &\equiv \sigma_q\beta_q \pmod{2\pi}.
    \label{eq:true-sign-q}
\end{align}
The sign branch $-\sigma$ corresponds to the reflected phase vector $(-\phi_1,\dots,-\phi_q)$, since all phase relations change sign under reflection.
Thus, the global sign change $\sigma\mapsto-\sigma$ is an unavoidable ambiguity: no reflection-invariant third-order moment can distinguish these two branches. The remaining question is whether any other sign branch is compatible with the third-order moment data. Below, we impose a generic nondegeneracy condition and prove in Lemma~\ref{lem:sign-branch-elimination} that the consistency constraints exclude all branches except $\sigma$ and $-\sigma$.

We use additional third-order moments to test whether a sign branch is compatible with the full reflection-invariant third-moment data. Define
\begin{align}
    d_j &:= \frac{M^{(3)}_{2,j,j+2}(\theta)}{2r_2r_jr_{j+2}} = \cos(\phi_2+\phi_j-\phi_{j+2}),
    \qquad 2\le j\le q-2,
    \label{eq:dj}
    \\
    d_\ast &:= \frac{M^{(3)}_{2,q-1,q}(\theta)}{2r_2r_{q-1}r_q} = \cos(\phi_2+\phi_{q-1}+\phi_q),
    \label{eq:dstar}
\end{align}
where~\eqref{eq:dstar} follows from $2+(q-1)+q=p$ and Proposition~\ref{prop:third-moment-cosines}(2).

Substituting the recursion relation formula~\eqref{eq:phi-explicit-anchored} into~\eqref{eq:dj}-\eqref{eq:dstar}, the anchor $\alpha=\phi_1$ cancels. Thus, the consistency relations depend only on the sign branch:
\begin{align}
    \phi_2+\phi_j-\phi_{j+2}
    &\equiv
    -\varepsilon_1\beta_1+\varepsilon_j\beta_j+\varepsilon_{j+1}\beta_{j+1}
    \pmod{2\pi},
    \qquad 2\le j\le q-2,
    \label{eq:interior-angle}
    \\
    \phi_2+\phi_{q-1}+\phi_q
    &\equiv -\varepsilon_1\beta_1+\varepsilon_{q-1}\beta_{q-1}+\varepsilon_q\beta_q \pmod{2\pi}.
    \label{eq:endpoint-angle}
\end{align}
Equivalently, any admissible sign branch must satisfy the sign-only constraints
\begin{align}
    \cos\bigl(-\varepsilon_1\beta_1+\varepsilon_j\beta_j+\varepsilon_{j+1}\beta_{j+1}\bigr)
    &=
    d_j,
    \qquad 2\le j\le q-2,
    \label{eq:sign-only-dj}
    \\
    \cos\bigl(-\varepsilon_1\beta_1+\varepsilon_{q-1}\beta_{q-1}+\varepsilon_q\beta_q\bigr)
    &=
    d_\ast.
    \label{eq:sign-only-dstar}
\end{align}
Thus, these constraints are used to prune incompatible sign branches. The cyclic-shift degree of freedom, represented by the anchor $\phi_1$, is recovered later from the terminal chain equation.

\begin{remark}[Choice of the consistency family]
\label{rem:choice-consistency-family}
The family $M^{(3)}_{2,j,j+2}$, together with the endpoint entry $M^{(3)}_{2,q-1,q}$, is chosen because it yields the local sign relations~\eqref{eq:interior-angle}-\eqref{eq:endpoint-angle}. Other nonzero third-moment entries give valid consistency checks as well, but typically involve longer combinations of signs after substituting the chain recursion. The family~\eqref{eq:dj}-\eqref{eq:dstar} is therefore a convenient choice for the constructive proof and the branch-and-prune procedure.
\end{remark}

We impose a nondegeneracy condition that rules out accidental collisions among the cosine consistency values. This ensures that the sign-only constraints identify the correct sign branch, up to the unavoidable global reflection.

\begin{assum}[Nondegeneracy and sign separation]
\label{ass:generic-sign-separation}
Let $p=2q+1\ge 7$. Recall the definition of $\beta_j$ in~\eqref{eq:beta-definition} and let $\sigma=(\sigma_1,\dots,\sigma_q)$ denote the true sign branch defined by~\eqref{eq:true-sign-j}-\eqref{eq:true-sign-q}. Assume that
\begin{align}
    \beta_j\notin \pi\mathbb{Z},
    \qquad j=1,\dots,q,
    \label{eq:beta-nondegenerate}
\end{align}
that the first consistency value separates the true initial sign triple up to global reflection, namely,
\begin{align}
    \cos\bigl(-\tau_1\beta_1+\tau_2\beta_2+\tau_3\beta_3\bigr) \neq \cos\bigl(-\sigma_1\beta_1+\sigma_2\beta_2+\sigma_3\beta_3\bigr),
    \label{eq:base-separation}
\end{align}
for every $(\tau_1,\tau_2,\tau_3)\in\{\pm1\}^3$ with $(\tau_1,\tau_2,\tau_3)\neq \pm(\sigma_1,\sigma_2,\sigma_3)$, and that
\begin{align}
    -\sigma_1\beta_1+\sigma_m\beta_m \notin \pi\mathbb{Z},
    \qquad m=3,\dots,q-1.
    \label{eq:generic-tail-separation}
\end{align}
\end{assum}

Condition~\eqref{eq:base-separation} makes the initial sign triple identifiable, up to the global sign change corresponding to reflection. Condition~\eqref{eq:generic-tail-separation} then ensures that each subsequent local consistency check separates the newly introduced sign.

\begin{remark}[Genericity of Assumption~\ref{ass:generic-sign-separation}]
\label{rem:genericity-sign-separation}
Assumption~\ref{ass:generic-sign-separation} holds outside a measure-zero set of phase vectors. Indeed, the condition $\beta_j\notin\pi\mathbb{Z}$ excludes the cases in which the corresponding chain angle is an integer multiple of $\pi$. For each fixed sign triple $(\tau_1,\tau_2,\tau_3)\neq \pm(\sigma_1,\sigma_2,\sigma_3)$, failure of~\eqref{eq:base-separation} is equivalent to a congruence of the form
\[
    -\tau_1\beta_1+\tau_2\beta_2+\tau_3\beta_3 \equiv \pm\bigl(-\sigma_1\beta_1+\sigma_2\beta_2+\sigma_3\beta_3\bigr) \pmod{2\pi},
\]
and hence defines a finite union of affine hyperplanes. The exclusions in~\eqref{eq:generic-tail-separation} are also linear resonance conditions modulo $\pi$. Since only finitely many such conditions are imposed, their union has Lebesgue measure zero. Thus, the assumption is generic in the sense that the degenerate phase configurations form a measure-zero exceptional set.
\end{remark}

\begin{lem}[Sign-branch elimination]
\label{lem:sign-branch-elimination}
Let $p=2q+1\ge 7$, and assume Assumption~\ref{ass:generic-sign-separation}. Let $\sigma=(\sigma_1,\dots,\sigma_q)$ denote the true sign branch and let $\varepsilon=(\varepsilon_1,\dots,\varepsilon_q)\in\{\pm1\}^q$ be a sign branch satisfying the sign-only consistency equations~\eqref{eq:sign-only-dj}-\eqref{eq:sign-only-dstar}. Then $\varepsilon=\pm\sigma$.
\end{lem}

\begin{proof}[Proof of Lemma~\ref{lem:sign-branch-elimination}]
The consistency equations are invariant under the global sign change $\varepsilon\mapsto -\varepsilon$, since every consistency angle changes sign and cosine is even. We may therefore assume, without loss of generality, that $\varepsilon_1=\sigma_1$.

We first identify the initial triple. If $q=3$, the endpoint constraint~\eqref{eq:sign-only-dstar} gives
\begin{align}
    \cos\bigl(-\sigma_1\beta_1+\varepsilon_2\beta_2+\varepsilon_3\beta_3\bigr) = \cos\bigl(-\sigma_1\beta_1+\sigma_2\beta_2+\sigma_3\beta_3\bigr).
    \label{eq:initial-triple-q3}
\end{align}
If $q\ge4$, the same identity follows from the interior constraint~\eqref{eq:sign-only-dj} with $j=2$. In either case, the base-separation condition~\eqref{eq:base-separation} implies
\begin{align}
    (\varepsilon_1,\varepsilon_2,\varepsilon_3) = (\sigma_1,\sigma_2,\sigma_3).
    \label{eq:initial-triple-identified}
\end{align}

It remains to show that all later signs agree. Suppose, for contradiction, that there is an index $j\ge4$ with $\varepsilon_j\neq\sigma_j$, and let $j$ be the smallest such index. If $j\le q-1$, then applying~\eqref{eq:sign-only-dj} with index $j-1$, and using the minimality of $j$, gives
\begin{align}
    \cos(x-\sigma_j\beta_j) = \cos(x+\sigma_j\beta_j),
    \qquad
    x:=-\sigma_1\beta_1+\sigma_{j-1}\beta_{j-1}.
    \label{eq:tail-sign-contradiction-interior}
\end{align}
If $j=q$, the endpoint constraint~\eqref{eq:sign-only-dstar} gives the same identity with
$x:=-\sigma_1\beta_1+\sigma_{q-1}\beta_{q-1}$. In both cases,
\begin{align}
    \cos(x-\sigma_j\beta_j)-\cos(x+\sigma_j\beta_j) = 2\sin(x)\sin(\sigma_j\beta_j) = 0.
    \label{eq:cos-difference-sign-proof}
\end{align}
But $\sin(\sigma_j\beta_j)\neq0$ by~\eqref{eq:beta-nondegenerate}, and $\sin(x)\neq0$ by~\eqref{eq:generic-tail-separation}. This is a contradiction.

Therefore no such $j$ exists, and hence $\varepsilon=\sigma$ under the normalization $\varepsilon_1=\sigma_1$. Removing this normalization gives $\varepsilon=\pm\sigma$.
\end{proof}

Lemma~\ref{lem:sign-branch-elimination} shows that the consistency constraints determine the sign branch uniquely, up to the unavoidable global reflection $\varepsilon\mapsto-\varepsilon$. Once this branch is identified, the terminal chain equation recovers the admissible values of $\phi_1$, and the remaining phases follow from the frequency-marching recursion.

\subsection{ Proof of Theorem~\ref{thm:projected-mra-anchored-recovery}}
\label{sec:proof_thm_projected-mra-anchored-recovery}

We now combine the preceding ingredients to prove the main recovery result, which is summarized in Algorithm~\ref{alg:projected-mra-constructive-recovery}.

By~\eqref{eq:M1-fourier}, the first population moment determines the zero Fourier mode $\hat{\theta}[0]$. Once $\hat{\theta}[0]$ is known, the affine relation~\eqref{eq:Xl-reduced-matrix-form} allows the observation-space moments up to order three to be converted into the projected Fourier-cosine moments up to order three by~\eqref{eq:transfer-projected-to-cosine-moments}, because the cosine matrix $A$ is invertible by Lemma~\ref{lem:cosine-matrix-invertible}. Hence the moments $T^{(2)}(\theta)$ and $T^{(3)}(\theta)$ determine the corresponding quantities $M^{(2)}_{k,m}(\theta)$ and $M^{(3)}_{a,b,c}(\theta)$.

Proposition~\ref{prop:second-moment-amplitudes} determines the magnitudes $r_k=|\hat{\theta}[k]|$, for $k=1,\dots,q$.
Proposition~\ref{prop:third-moment-cosines} then determines the chain cosines~\eqref{eq:cj}-\eqref{eq:cq} and the consistency cosines~\eqref{eq:dj}--\eqref{eq:dstar}. Therefore the values $\beta_1,\dots,\beta_q$ in~\eqref{eq:beta-definition} and the sign-only consistency equations~\eqref{eq:sign-only-dj}--\eqref{eq:sign-only-dstar} are determined by the first three moments. By Lemma~\ref{lem:sign-branch-elimination}, these equations identify the true sign branch up to the global reflection ambiguity $\varepsilon\mapsto-\varepsilon$. Fix one of the two branches, denoted $\varepsilon$.

It remains to recover the phases. By Lemma~\ref{lem:anchored-frequency-marching},
\begin{align}
    \phi_k \equiv k\phi_1-\sum_{t=1}^{k-1}\varepsilon_t\beta_t \pmod{2\pi},
    \qquad k=1,\dots,q.
    \label{eq:phase-recursion-main-proof}
\end{align}
Substituting $k=q$ into the terminal chain equation
$\phi_1+2\phi_q=\varepsilon_q\beta_q$ gives
\begin{align}
    p\phi_1 \equiv 2\sum_{t=1}^{q-1}\varepsilon_t\beta_t+\varepsilon_q\beta_q \pmod{2\pi}.
    \label{eq:phi1-recovery-main-proof}
\end{align}
Thus $\phi_1$ is determined modulo $2\pi/p$, giving exactly $p$ admissible values. For each such value,~\eqref{eq:phase-recursion-main-proof} determines all phases $\phi_2,\dots,\phi_q$, and hence all Fourier coefficients $\hat{\theta}[k]=r_ke^{i\phi_k}$, with $\hat{\theta}[-k]=\overline{\hat{\theta}[k]}$, together with the already recovered coefficient $\hat{\theta}[0]$. Applying the inverse Fourier transform gives $p$ signal representatives.

Finally, these $p$ representatives are exactly the cyclic-shift orbit. Indeed, $\widehat{R_\ell\theta}[k] = e^{-2\pi i k\ell/p}\hat{\theta}[k]$, so a shift changes $\phi_1$ by $-2\pi\ell/p$. The remaining global sign ambiguity in the branch corresponds to reflection, since $J$ sends $\hat{\theta}[k]\mapsto \overline{\hat{\theta}[k]}$, and $\phi_k\mapsto-\phi_k$.
Consequently, the recovered representatives are precisely $D_{2p}\cdot\theta = \{R_\ell\theta,\;R_\ell J\theta:\ell\in\mathbb{Z}/p\}$.
This proves the theorem.

\begin{algorithm}[t!]
\caption{Signal recovery from population moments}
\label{alg:projected-mra-constructive-recovery}
\textbf{Input:} Moments $T^{(1)}(\theta),T^{(2)}(\theta),T^{(3)}(\theta)$.

\textbf{Output:} An estimate of the dihedral orbit $D_{2p}\cdot\theta$.

\textbf{Procedure:}
\begin{enumerate}
    \item Recover $\hat\theta[0]$ from $T^{(1)}(\theta)$ using~\eqref{eq:M1-fourier}, and convert the  moments $T^{(2)}(\theta)$ and $T^{(3)}(\theta)$ to $M^{(2)}(\theta)$ and $M^{(3)}(\theta)$ using~\eqref{eq:transfer-projected-to-cosine-moments}.

    \item Recover the magnitudes $r_1,\ldots,r_q$ from $M^{(2)}(\theta)$ by Proposition~\ref{prop:second-moment-amplitudes}.

    \item Extract the chain cosines $c_1,\ldots,c_q$ from~\eqref{eq:cj}-\eqref{eq:cq}, set $\beta_j=\arccos(c_j)$, and extract the consistency cosines $d_j$ and $d_\ast$ from~\eqref{eq:dj}-\eqref{eq:dstar}.

    \item Recover the sign branch sequentially. Fix a reflection normalization,  $\varepsilon_1=+1$. Use the first consistency relation, \eqref{eq:sign-only-dj} with $j=2$, to determine $(\varepsilon_1,\varepsilon_2,\varepsilon_3)$. Then, for $m=4,\ldots,q-1$, determine $\varepsilon_m$ from $\cos(-\varepsilon_1\beta_1+\varepsilon_{m-1}\beta_{m-1}+\varepsilon_m\beta_m)=d_{m-1}$, and determine $\varepsilon_q$ from $\cos(-\varepsilon_1\beta_1+\varepsilon_{q-1}\beta_{q-1}+\varepsilon_q\beta_q)=d_\ast$. 

    \item Let $\widehat\varepsilon=(\widehat\varepsilon_1,\ldots,\widehat\varepsilon_q)$ be the recovered normalized branch. The two surviving branches are $\widehat\varepsilon$ and $-\widehat\varepsilon$, corresponding to reflection.

    \item Recover the admissible values of $\phi_1$ from~\eqref{eq:phi1-recovery-main-proof}. For each such value, recover the remaining phases by~\eqref{eq:phi-explicit-anchored}, form $\hat\theta[k]=r_ke^{i\phi_k}$ and $\hat\theta[-k]=\overline{\hat\theta[k]}$ for $k=1,\ldots,q$, and invert the Fourier transform.
\end{enumerate}
\end{algorithm}

\paragraph{Computational complexity.} The constructive recovery procedure in Algorithm~\ref{alg:projected-mra-constructive-recovery} is computationally efficient. From $n$ samples $\{y_i \}_{i = 1}^n$ in $\mathbb{R}^q$, the empirical moments up to order three can be formed directly in $O(nq^3)$ operations. The transfer from the moments of the projected MRA observation to Fourier-cosine moments, given by~\eqref{eq:transfer-projected-to-cosine-moments}, can be implemented by applying the change of coordinates along the tensor modes, at cost $O(q^4)$ for the full third-order tensor.

The main computational point is that the phase-recovery step is not exponential in $q$. The cosine phase equations generate $2^q$ possible sign branches, and a direct implementation analogous to the sign search in dihedral MRA~\cite{edidin2026reflection} would require an exhaustive search over these branches. In contrast, the proof of Lemma~\ref{lem:sign-branch-elimination} gives a sequential branch-and-prune procedure: after fixing the global reflection ambiguity, the first consistency relation determines the initial sign triple, and each subsequent local consistency relation determines one additional sign. Hence, the number of active normalized branches remains one throughout the algorithm, and sign recovery costs $O(q)$ rather than $O(2^q)$.

Following the recovery of the sign branch, the phase recursion costs $O(q)$, and the final inverse Fourier transform costs $O(p\log p)$ using an FFT. Therefore, using the direct implementations described above, the overall finite-sample complexity is
\begin{align}
    O(nq^3) + O(q^4) + O(q) + O(p\log p) = O(nq^3 + q^4).
\end{align}
This is polynomial in the signal dimension, linear in $n$, and, in particular, avoids the exponential $O(2^q)$ sign search that would arise from a direct enumeration of all phase branches.

\section{Finite-sample implementation and numerical experiments}
\label{sec:finite-sample-analysis}

Although the population proof of Theorem~\ref{thm:projected-mra-anchored-recovery} suggests an explicit branch-and-prune reconstruction procedure (Algorithm~\ref{alg:projected-mra-constructive-recovery}), we do not use this algorithm as the main finite-sample estimator. The reason is that the proof uses only a selected subset of the available third-order moment equations, chosen to give a transparent constructive identifiability argument. At the population level, this loss of information is harmless, since the selected equations already determine the dihedral orbit generically. In finite samples, however, the branch decisions rely on exact algebraic equalities among cosine relations, and these equalities are perturbed by sampling noise and by the transformation from  moments of the projected MRA to Fourier-cosine coordinates. Consequently, a hard pruning rule can incorrectly discard the correct branch, while phase propagation along a single chain can accumulate error. This makes the branch-and-prune procedure numerically fragile in finite samples. We therefore use it primarily as a population identifiability mechanism, and compare more robust finite-sample estimators based on global objective functions that use all the moment information.

We compare three reconstruction strategies: expectation--maximization (EM) as a likelihood-based benchmark, direct optimization of the  moments $T^{(d)}$~\eqref{eq:population-moment-definition}, and direct optimization of the transformed Fourier-cosine moments $M^{(d)}$~\eqref{eq:cosine-population-moment-definition}. 

\subsection{Empirical moments and Gaussian debiasing}

The population recovery procedure is formulated in terms of the noiseless moments of the projected MRA signal, $T^{(1)}(\theta),T^{(2)}(\theta),T^{(3)}(\theta)$. In the finite-sample setting, these moments are not observed directly. Instead, we observe the samples
\begin{align}
    y_i=\Pi(R_{\ell_i}\theta)+\xi_i \; \in \; \mathbb{R}^q,
    \qquad i=1,\dots,n,
    \label{eq:finite-sample-observation-model}
\end{align}
where $\ell_i$ are i.i.d. uniform on $\mathbb{Z}/p$, $\Pi$ is the projection map~\eqref{eq:projection-map-model}, and $\xi_i\sim\mathcal{N}(0,\sigma^2 I_q)$ are independent Gaussian noise vectors.
In this section, the ground-truth signal $\theta\in\mathbb{R}^p$ is normalized to a unit Euclidean norm.

We estimate the noiseless moments in two steps. First, we form raw empirical moments from the noisy observations. Second, we subtract the known Gaussian noise contribution to obtain unbiased estimators of the corresponding noiseless moments.

\begin{definition}[Raw empirical moments]
\label{def:empirical-moments}
For $d=1,2,3$, define the raw empirical $d$-th moment by
\begin{align}
    \widehat{T}_n^{(d)} := \frac{1}{n}\sum_{i=1}^n y_i^{\otimes{d}} \in (\mathbb{R}^q)^{\otimes{d}}.
    \label{eq:empirical-moments-general}
\end{align}
Thus, $\widehat{T}_n^{(1)}\in\mathbb{R}^q$, $\widehat{T}_n^{(2)}\in\mathbb{R}^{q\times q}$, and $\widehat{T}_n^{(3)}\in\mathbb{R}^{q\times q\times q}$.
\end{definition}

Since the observations contain additive Gaussian noise, the raw empirical moments are biased estimators of the noiseless signal moments. We therefore define the following debiased moments.

\begin{definition}[Gaussian-debiased empirical moments]
\label{def:corrected-empirical-moments}
Let $\widehat{\mu}:=\widehat{T}_n^{(1)}\in\mathbb{R}^q$. Define
\begin{align}
    \widetilde{T}_n^{(1)} &:= \widehat{T}_n^{(1)} \in\mathbb{R}^q,
    \label{eq:corrected-empirical-moment-1}
    \\    \widetilde{T}_n^{(2)} &:= \widehat{T}_n^{(2)}-\sigma^2 I_q \in\mathbb{R}^{q\times q},
    \label{eq:corrected-empirical-moment-2}
\end{align}
and define $\widetilde{T}_n^{(3)}\in\mathbb{R}^{q\times q\times q}$ entrywise by
\begin{align}
    (\widetilde{T}_n^{(3)})_{abc}
    &:= (\widehat{T}_n^{(3)})_{abc} - \sigma^2\bigl( \widehat{\mu}_a\delta_{bc} + \widehat{\mu}_b\delta_{ac} + \widehat{\mu}_c\delta_{ab} \bigr),
    \qquad 1\le a,b,c\le q.
    \label{eq:corrected-empirical-moment-3}
\end{align}
\end{definition}

The next lemma records the noisy population moments and shows that $\widetilde{T}_n^{(1)},\widetilde{T}_n^{(2)},\widetilde{T}_n^{(3)}$ are unbiased and consistent estimators of the noiseless moments.

\begin{lem}[Consistency of Gaussian-debiased moments]
\label{lem:corrected-moments-consistency}
Let $Y=\Pi(R_\ell\theta)+\xi$, where $\ell\sim\mathrm{Unif}(\mathbb{Z}/p)$ and $\xi\sim\mathcal{N}(0,\sigma^2 I_q)$ are independent of $\ell$. Denote $\mu:= \mathbb{E}[Y]$. Then
\begin{align}
    \mathbb{E}[Y] &= T^{(1)}(\theta),
    \label{eq:noisy-moment-1}
    \\
    \mathbb{E}[Y^{\otimes2}] &= T^{(2)}(\theta)+\sigma^2 I_q,
    \label{eq:noisy-moment-2}
    \\
    \mathbb{E}[Y_aY_bY_c] &= T^{(3)}_{abc}(\theta) + \sigma^2\bigl( \mu_a\delta_{bc} + \mu_b\delta_{ac} + \mu_c\delta_{ab}\bigr).
    \label{eq:noisy-moment-3}
\end{align}
Consequently, for $d=1,2,3$,
\begin{align}
    \mathbb{E}[\widetilde{T}_n^{(d)}] = T^{(d)}(\theta),
    \label{eq:corrected-moments-unbiased}
\end{align}
and, as $n\to\infty$, $\widetilde{T}_n^{(d)} \xrightarrow{\mathrm{a.s.}} T^{(d)}(\theta)$.
\end{lem}

\begin{proof}[Proof of Lemma~\ref{lem:corrected-moments-consistency}]
Write $Y=X+\xi$, where $X=\Pi(R_\ell\theta)$. Since $\xi$ is independent of $X$, has mean zero, and has covariance $\sigma^2 I_q$, we immediately obtain \eqref{eq:noisy-moment-1} and~\eqref{eq:noisy-moment-2}. For the third moment, expanding $(X_a+\xi_a)(X_b+\xi_b)(X_c+\xi_c)$ and using the Gaussian identities $\mathbb{E}[\xi_a]=0$, $\mathbb{E}[\xi_a\xi_b]=\sigma^2\delta_{ab}$, and $\mathbb{E}[\xi_a\xi_b\xi_c]=0$, gives~\eqref{eq:noisy-moment-3}.
The unbiasedness identities in~\eqref{eq:corrected-moments-unbiased} follow directly from~\eqref{eq:noisy-moment-1}--\eqref{eq:noisy-moment-3} and Definition~\ref{def:corrected-empirical-moments}. Finally, the raw empirical moments converge almost surely, entrywise, by the strong law of large numbers, since the required moments are finite.
\end{proof}

\subsection{Benchmark algorithms}
\label{subsec:benchmark-algorithms}

We compare two moment-based reconstruction procedures with a likelihood-based benchmark. The first benchmark is expectation-maximization (EM) for the projected MRA likelihood. The second method is direct moment optimization, which minimizes a least-squares discrepancy between the debiased empirical moments $\widetilde{T}_n^{(d)}$ and their model-predicted counterparts ${T}^{(d)}$. The third method is direct Fourier-cosine moment optimization, which applies the same optimization principle after transforming the empirical moments to the Fourier-cosine coordinates ${M}^{(3)}$. These comparisons separate three effects: likelihood-based estimation, moment matching in the original projected coordinates, and moment matching in the Fourier-cosine coordinates that expose the algebraic magnitude and phase relations.

\subsubsection{Expectation-maximization}
We apply the EM algorithm to the projected MRA model, defined in Problem~\ref{prob:projected-mra} with latent shifts $\ell_i$. For a current iterate $\theta^{(t)}$, the E-step computes the posterior responsibility of each shift:
\begin{align}
    w_{i,\ell}^{(t)}
    &:=
    \mathbb{P}(\ell_i=\ell\mid y_i,\theta^{(t)})
    \notag\\
    &=
    \frac{\exp\!\left(-\frac{1}{2\sigma^2} \bigl\|y_i-\Pi(R_\ell\theta^{(t)})\bigr\|_F^2 \right)}{\sum_{m\in\mathbb{Z}/p} \exp\!\left( -\frac{1}{2\sigma^2} \bigl\|y_i-\Pi(R_m\theta^{(t)})\bigr\|_F^2 \right)},
    \label{eq:projected-mra-em-estep}
\end{align}
where $\| \cdot \|$ denotes the Frobenius norm. The M-step updates $\theta$ by minimizing the expected squared reconstruction error,
\begin{align}
    \theta^{(t+1)} \in \argmin_{\theta\in\mathbb{R}^p} \sum_{i=1}^n \sum_{\ell\in\mathbb{Z}/p} w_{i,\ell}^{(t)} \bigl\|y_i-\Pi(R_\ell\theta)\bigr\|_F^2.
    \label{eq:projected-mra-em-mstep}
\end{align}
Since the objective in~\eqref{eq:projected-mra-em-mstep} is quadratic in $\theta$, the
M-step is a linear least-squares problem,
and the solution is given by 
\begin{align}
    \theta^{(t+1)} = \left(\sum_{i=1}^n\sum_{\ell\in\mathbb{Z}/p} w_{i,\ell}^{(t)} R_\ell^\top \Pi^\top \Pi R_\ell \right)^{-1}\sum_{i=1}^n\sum_{\ell\in\mathbb{Z}/p} w_{i,\ell}^{(t)} R_\ell^\top \Pi^\top y_i .
    \label{eq:projected-mra-em-normal-equations}
\end{align}
In the numerical experiments, EM is run from $S_{\mathrm{EM}} = 5$ different initializations, and the estimate with the largest final log-likelihood is retained. We denote by $I_{\mathrm{EM}}$ the number of EM iterations until convergence, or the prescribed maximum number of iterations if convergence is not reached. Because the likelihood is invariant under the dihedral ambiguity of the projected model, the EM estimate is evaluated only up to the dihedral orbit. The dominant cost in the E-step is the evaluation of $np$ shifted projections in dimension $q$, which costs $O(npq)=O(nq^2)$. The M-step is a linear least-squares problem in $\theta\in\mathbb{R}^p$; in a straightforward implementation, solving the corresponding $p\times p$ linear system costs $O(p^3)=O(q^3)$ per iteration. Thus, a straightforward implementation with $S_{\mathrm{EM}}$ initializations and $I_{\mathrm{EM}}$ iterations costs $O\!\left(S_{\mathrm{EM}}I_{\mathrm{EM}}(nq^2+q^3)\right)$. In the regimes considered here, the E-step typically dominates when $n$ is large compared with $q$.

\subsubsection{Direct moment optimization}
The direct $T$-moment baseline minimizes the discrepancy between the debiased empirical moments and their model-predicted counterparts. Specifically, using the debiased empirical moments $\widetilde{T}_n^{(d)}$ from Definition~\ref{def:corrected-empirical-moments} and the population moments $T^{(d)}$ defined in~\eqref{eq:population-moment-definition}, we consider the block-normalized least-squares objective
\begin{align}
    \mathcal{L}_T(\vartheta) := \sum_{d=1}^3 \frac{\|T^{(d)}(\vartheta)-\widetilde{T}_n^{(d)}\|_F^2} {\|\widetilde{T}_n^{(d)}\|_F^2}.
    \label{eq:direct-T-moment-loss-expanded}
\end{align}
The normalization balances the contributions of the different moment orders, whose raw scales may differ substantially. We then define
\begin{align}
    \widehat{\theta}_T \in \argmin_{\vartheta\in\mathbb{R}^p} \mathcal{L}_T(\vartheta).
    \label{eq:direct-T-moment-loss}
\end{align}

In implementation, the three normalized residual blocks in~\eqref{eq:direct-T-moment-loss-expanded} are stacked into a single residual vector, and the resulting nonlinear least-squares problem is solved using MATLAB's \texttt{lsqnonlin} routine. The optimization is run from $S_T$ initializations, and the solution with the smallest final value of $\mathcal{L}_T$ is retained. The empirical moments $\widehat{T}^{(1)},\widehat{T}^{(2)},\widehat{T}^{(3)}$ are formed once from the data; computing them directly from $n$ observations in $\mathbb{R}^q$ costs $O(nq^3)$, dominated by the third moment. For each candidate signal evaluated during the optimization, computing the population moments $T^{(1)},T^{(2)},T^{(3)}$ by averaging over all $p$ shifts costs $O(pq^3)=O(q^4)$. Thus, if $I_T$ denotes the number of nonlinear least-squares iterations, the total computational cost of the direct moment method is approximately
\begin{align}
    O(nq^3) + O(S_T I_T q^4).
\end{align}

\subsubsection{Direct Fourier-cosine moment optimization}
The direct $M$-moment baseline performs the analogous optimization in the Fourier-cosine coordinates. Let $\widehat{M}^{(2)}$ and $\widehat{M}^{(3)}$ be the transformed empirical moments obtained by applying~\eqref{eq:transfer-projected-to-cosine-moments} to the debiased empirical moments from Definition~\ref{def:corrected-empirical-moments}, and let $M^{(2)}(\vartheta)$ and $M^{(3)}(\vartheta)$ be the corresponding model-predicted Fourier-cosine moments, as defined in~\eqref{eq:cosine-population-moment-definition}. We solve
\begin{align}
    \widehat{\theta}_M \in \argmin_{\vartheta\in\mathbb{R}^p} \mathcal{L}_M(\vartheta),
    \label{eq:direct-M-moment-loss}
\end{align}
where
\begin{align}
    \mathcal{L}_M(\vartheta)
    &:= \frac{\|T^{(1)}(\vartheta)-\widetilde{T}_n^{(1)}\|_F^2} {\|\widetilde{T}_n^{(1)}\|_F^2} + \frac{\|M^{(2)}(\vartheta)-\widehat{M}^{(2)}\|_F^2} {\|\widehat{M}^{(2)}\|_F^2} + \frac{\|M^{(3)}(\vartheta)-\widehat{M}^{(3)}\|_F^2} {\|\widehat{M}^{(3)}\|_F^2}.
    \label{eq:direct-M-moment-loss-expanded}
\end{align}
The first-order term is kept in projected coordinates because it determines the zero-frequency component directly, while the second- and third-order terms are evaluated in the Fourier-cosine coordinates that reveal the magnitude and phase relations. As for direct $T$-optimization, the normalized residual blocks in~\eqref{eq:direct-M-moment-loss-expanded} are stacked into a single residual vector and optimized using MATLAB's \texttt{lsqnonlin} routine. The optimization is run from $S_M$ initializations, and the solution with the smallest final value of $\mathcal{L}_M$ is retained. We denote by $I_M$ the number of nonlinear least-squares iterations used by this procedure.

In our implementation, the empirical Fourier-cosine moments are computed once from the data. The preprocessing stage costs $O(nq^3)+O(q^4)$: the first term comes from forming the empirical moments up to order three, and the second from transferring the full third-order tensor to Fourier-cosine coordinates using~\eqref{eq:transfer-projected-to-cosine-moments}. During optimization, the model-predicted Fourier-cosine moments $M^{(2)}(\vartheta)$ and $M^{(3)}(\vartheta)$ are evaluated as full moment tensors. Consequently, each residual evaluation costs $O(q^4)$, comparable to the direct moment objective. Including the $S_M$ starts and $I_M$ nonlinear least-squares iterations, the overall cost is
\begin{align}
    O(nq^3) + O(S_M I_M q^4).
\end{align}

\subsubsection{Cost summary}
All moment-based methods first require forming empirical moments from the data. With a direct implementation, the empirical moments up to order three in the projected coordinates cost $O(nq^3)$ to form, dominated by the third-order moment. Thus, the total costs in our implementation are
\[
    \begin{array}{c|c}
    \text{Method} & \text{Total cost} \\
    \hline
    \text{Direct }T\text{-moment optimization}
    & O(nq^3) + O(S_T I_T q^4) \\
    \text{Direct }M\text{-moment optimization}
    & O(nq^3) + O(S_M I_M q^4) \\
    \text{EM}
    & O(S_{\mathrm{EM}} I_{\mathrm{EM}} n q^2)
    \end{array}
\]
Here, $S_T,S_M,S_{\mathrm{EM}}$ denote the number of starts, and $I_T,I_M,I_{\mathrm{EM}}$ denote the corresponding numbers of iterations. The key distinction is that moment-based methods pay the sample-size-dependent cost only once, when the empirical moments are formed, while EM continues to scale linearly with $n$ at every iteration.

\begin{remark}
Empirically, Algorithm~\ref{alg:projected-mra-constructive-recovery} is less stable than direct $M$-moment optimization, which is itself less stable than direct $T$-moment optimization. The latter instability is consistent with the fact that converting the projected moments $T^{(d)}$ into Fourier-cosine moments $M^{(d)}$ requires applying $A^{-1}$ along the tensor modes, which amplifies finite-sample fluctuations. Thus, although the Fourier-cosine representation is useful for proving population identifiability, it is not necessarily the most stable representation for finite-sample estimation.

A natural future direction is therefore to separate two issues: the stability of the reflection-invariant phase-recovery mechanism itself, and the additional conditioning loss introduced by the projected-MRA change of coordinates. The former can be studied more directly in the dihedral orbit-recovery problem, where the same cosine phase-coupling relations arise without first transferring projected moments through $A^{-1}$. We leave this direction to future work.
\end{remark}

\subsection{Numerical results}
\label{subsec:numerical-results}

We now evaluate the finite-sample performance of the three reconstruction methods described above. In all experiments, we fix $p=13$, so that $q=6$, generate a generic unit-norm signal $\theta^\star$, and draw $n=2\times10^4$ observations from the finite-sample projected MRA model for each noise level $\sigma$. We use $20$ logarithmically spaced noise levels in the range $\sigma\in[0.05,1]$, and for each value of $\sigma$ we run $100$ independent Monte Carlo trials. The Gaussian-debiased moments are formed as in Definition~\ref{def:corrected-empirical-moments}. For the $M$-moment objective, these debiased moments are further transformed to Fourier-cosine coordinates, yielding $\widehat{M}^{(2)}$ and $\widehat{M}^{(3)}$.

For both nonlinear moment objectives, the residual blocks are optimized using MATLAB's \texttt{lsqnonlin} routine with the Levenberg-Marquardt algorithm. Each direct moment optimization is run from $S_T=S_M=20$ initializations. All starts are run to completion, and the estimate with the smallest final residual norm is retained. EM is run from $S_{\mathrm{EM}}=5$ random initializations, and the estimate with the largest final likelihood is retained. For the nonlinear least-squares solvers, we use a maximum of $300$ iterations and $3000$ function evaluations, with function and step tolerances set to $10^{-10}$. For EM, we use a maximum of $2000$ iterations and stop earlier if the relative change between successive iterations falls below $10^{-8}$.

Since the signal is identifiable only up to the dihedral action, reconstruction accuracy is measured by the orbit-aligned error
\begin{align}
    d_{\mathrm{err}}(\widehat{\theta},\theta^\star) := \min_{g\in D_{2p}}\|\widehat{\theta}-g\cdot\theta^\star\|_2 .
    \label{eq:orbit-error}
\end{align}
To diagnose the effect of the Fourier-cosine transformation on empirical noise, we compare the third-order projected and transformed empirical moments with their noiseless population counterparts:
\[
    d_{\mathrm{Err},T_3} := \frac{\|\widetilde{T}_n^{(3)}-T^{(3)}(\theta^\star)\|_F} {\|T^{(3)}(\theta^\star)\|_F},
    \qquad    d_{\mathrm{Err},M_3} := \frac{\|\widehat{M}^{(3)}-M^{(3)}(\theta^\star)\|_F} {\|M^{(3)}(\theta^\star)\|_F}.
\]

Figure~\ref{fig:3}(a)-(b) compares the reconstruction accuracy of the three methods as a function of the noise level. EM achieves the smallest error over most of the range, as expected from a likelihood-based method that uses the full sample distribution. 
The dashed reference slopes in Figure~\ref{fig:3}(a)-(b) indicate two observed noise-scaling regimes.
In the low-noise regime, the EM mean-squared error scales as $\propto \sigma^2$, matching the classical scaling for a non-coherent estimation problem. As the noise level increases, the problem enters the low-SNR regime, where the mean-squared error of both EM and the moment-based estimators scales  as $\propto \sigma^6$. This agrees with the theoretical prediction for the projected MRA model: the phase information needed to identify the dihedral orbit first appears through third-order moment relations, so the information-theoretic sample-complexity exponent is governed by order three, giving the effective scaling $n/\sigma^6\to\infty$ for consistent recovery, as stated in Corollary~\ref{cor:projected-mra-sample-complexity}.

In this high-noise regime, the EM and the direct optimization of the $T$-moment become closer to each other, suggesting that the moments retain much of the statistically relevant information captured by the likelihood. By contrast, direct $M$-moment optimization deteriorates earlier. This does not reflect a loss of information in the Fourier-cosine coordinates at the population level; rather, it reflects finite-sample conditioning. The empirical transformation from the moments of the projected MRA observation to the Fourier-cosine representation $M$-moments applies $A^{-1}$ in tensor modes and therefore amplifies sampling fluctuations, making the transformed moment objective less stable at higher noise levels.

Figure~\ref{fig:3}(c) reports the median runtime. The two moment-based methods first form empirical moments and then solve nonlinear least-squares problems; after moment formation, their reconstruction costs are independent of $n$. Moreover, their runtime is nearly insensitive to the noise level in the range considered, since the empirical objectives have a fixed dimension and are optimized using a comparable number of nonlinear least-squares iterations across $\sigma$. In contrast, each EM iteration uses the full sample and therefore scales linearly with $n$. In addition, the number of EM iterations increases with the noise level: as $\sigma$ grows, the likelihood surface becomes flatter and the posterior distribution over shifts is less concentrated. The observed runtime growth is consistent with the iteration-complexity scaling of EM in the low-SNR regime, which behaves as $\sigma^4$ in the analysis of~\cite{balanov2026expectation}.

Finally, Figure~\ref{fig:3}(d) explains the gap between the two moment-based methods. At the population level, the projected and Fourier-cosine moment representations are equivalent. In finite samples, however, the transformation from $\widetilde{T}_n^{(3)}$ to $\widehat{M}^{(3)}$ applies $A^{-1}$ in each tensor mode. Thus,  empirical fluctuations in the projected third moment are transformed by three copies of $A^{-1}$, leading to a substantially larger relative error in $\widehat{M}^{(3)}$. This noise amplification explains why direct $T$-moment optimization can outperform direct $M$-moment optimization, even though the Fourier-cosine coordinates are the natural coordinates for the population identifiability proof.

\begin{figure}[t!]
    \centering
    \includegraphics[width=0.9\linewidth]{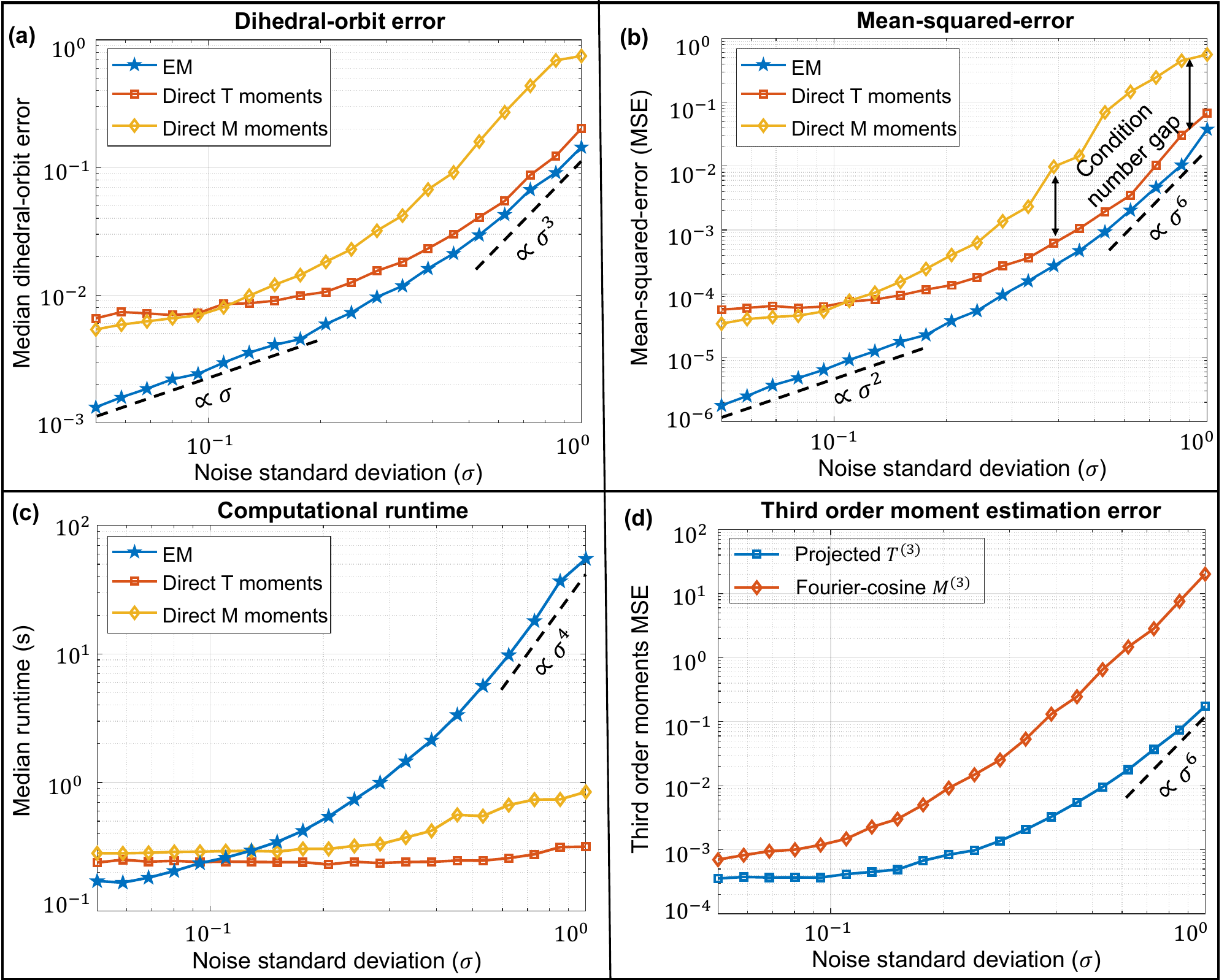}
    \caption{
    \textbf{Accuracy, runtime, and conditioning of finite-sample reconstruction.}
    We compare expectation--maximization (EM), direct moment optimization $T^{(d)}$, and  Fourier-cosine moment optimization $M^{(d)}$ for the projected MRA model with $p=13$, $q=6$, and $n=2\times10^4$. Reconstruction errors are computed up to the dihedral ambiguity.
    \textbf{(a)} Median reconstruction error and
    \textbf{(b)} mean-squared reconstruction error as functions of $\sigma$. EM and direct $T$-moment optimization achieve the best accuracy, while direct $M$-moment optimization degrades earlier as the transformed moment equations become noisier. 
    \textbf{(c)} Median runtime. EM uses the full likelihood at every iteration, while the two moment-based methods solve nonlinear least-squares problems after empirical moment formation.
    \textbf{(d)} Relative third-order moment error before and after transforming to Fourier-cosine coordinates. The empirical $M^{(3)}$ error is substantially larger than the projected $T^{(3)}$ error because the transformation applies $A^{-1}$ in each tensor mode. This finite-sample noise amplification explains the gap between $T$- and $M$-based moment methods despite their equivalence at the population level.
    }
    \label{fig:3}
\end{figure}

\section{Discussion and outlook}
\label{sec:discussion}

\subsection{A continuous analogue: bandlimited projected MRA on \texorpdfstring{$\mathrm{SO}(2)$}{SO(2)}}
\label{subsec:continuous-so2-extension}

The projected MRA model studied in this paper is formulated on the finite cyclic group $\mathbb{Z}/p$. A natural continuous analogue is obtained by replacing the cyclic shifts by rotations in $\mathrm{SO}(2)$. In such a setting, the unknown signal may be modeled as a real-valued function $f:\mathbb{S}^1\to\mathbb{R}$ on the circle, with observations of the form
\begin{align}
    Y = \Pi(R_\alpha f) + \xi,
    \qquad \alpha\sim\mathrm{Unif}([0,2\pi)),    \label{eq:continuous-projected-mra}
\end{align}
where $R_\alpha f(x)=f(x-\alpha)$, $\Pi$ is a reflection-symmetrizing projection, 
\begin{align}
    (\Pi f)(x) = f(x)+f(-x),
\end{align}
and $\xi$ is additive noise. 
The projection removes the orientation of the circle: rotating clockwise or counterclockwise produces the same projected orbit. Thus, although the latent transformations in~\eqref{eq:continuous-projected-mra} belong to $\mathrm{SO}(2)$, the identifiable ambiguity enlarges to the full orthogonal group $\mathrm{O}(2)$, generated by rotations and reflection.

To make the continuous problem finite-dimensional and computationally tractable, and to place it within the finite-dimensional framework assumed by standard sample-complexity results, one should assume that the functions are band-limited.
Namely, assume that
\begin{align}
    f(x) = \sum_{k=-K}^{K} \hat{f}_k e^{ikx},
    \qquad
    \hat{f}_{-k}=\overline{\hat{f}_k},    \label{eq:bandlimited-so2-signal}
\end{align}
for some finite bandwidth $K$. Under rotation, the Fourier coefficients transform as $\widehat{R_\alpha f}_k = e^{-ik\alpha}\hat{f}_k$.
After reflection symmetrization, each nonzero angular frequency contributes through the real-valued cosine-side coefficient
\begin{align}
    C_\alpha[k] = \hat{f}_k e^{-ik\alpha} + \hat{f}_{-k} e^{ik\alpha} = 2\Re\{\hat{f}_k e^{-ik\alpha}\},
    \qquad k=1,\ldots,K. \label{eq:continuous-cosine-coefficients}
\end{align}
Thus, the role of the projected Fourier-cosine coefficients $C_\ell[k]$ in the finite model~\eqref{eq:Cl-definition} is played by the angular Fourier-cosine coefficients $C_\alpha[k]$.

Writing $\hat{f}_k=r_ke^{i\phi_k}$, the second and third population moments of these coefficients obey the same selection rules as in the finite cyclic case, except that frequencies are now ordinary integers rather than elements of $\mathbb{Z}/p$. Averaging over $\alpha$ eliminates all terms whose total angular frequency is nonzero. Consequently,
\begin{align}
    M^{(2)}_{k,m} = \mathbb{E}_\alpha[C_\alpha[k]C_\alpha[m]] = 2r_k^2\delta_{km},
    \label{eq:so2-second-moment}
\end{align}
which is similar to Proposition~\ref{prop:second-moment-amplitudes}, so the second moment recovers the nonzero Fourier magnitudes. Similarly, analogously to Proposition~\ref{prop:third-moment-cosines}, whenever $a+b=c$,
\begin{align}
    M^{(3)}_{a,b,c} = \mathbb{E}_\alpha[C_\alpha[a]C_\alpha[b]C_\alpha[c]] = 2r_ar_br_c\cos(\phi_a+\phi_b-\phi_c).
    \label{eq:so2-third-moment}
\end{align}
Thus, the third moments again contain reflection-invariant cosine phase couplings.

The resulting phase-recovery mechanism is closely parallel to the finite projected MRA proof. The chain entries $M^{(3)}_{1,j,j+1}$, $1\le j\le K-1$, give $\cos(\phi_1+\phi_j-\phi_{j+1})$, and hence determine equations of the form
\begin{align}
    \phi_1+\phi_j-\phi_{j+1} \equiv \pm\beta_j \pmod{2\pi},
\end{align}
similar to the equations obtained in~\eqref{eq:chain-equation-j}-\eqref{eq:chain-equation-q}
After fixing a rotation gauge, for example $\phi_1=0$, these equations propagate the phases from low to high frequencies up to sign choices. Additional third-order relations, obtained from $M^{(3)}_{2,j,j+2}$, give consistency constraints of the form
\[
    \cos(-\varepsilon_1\beta_1+\varepsilon_j\beta_j+\varepsilon_{j+1}\beta_{j+1})=d_j,
\]
similar to~\eqref{eq:sign-only-dj}-\eqref{eq:sign-only-dstar} which are used to prune the sign choices. Under generic sign-separation assumptions analogous to those used in the cyclic proof (Assumption~\ref{ass:generic-sign-separation}), the sign branch are determined up to the global sign flip $\varepsilon\mapsto-\varepsilon$, corresponding to reflection of $f$.

We note that there is one important difference from the finite cyclic model. In the projected MRA problem on $\mathbb{Z}/p$, the wrapping relation $1+q+q=p$ supplies a terminal equation that recovers the anchor phase $\phi_1$ modulo the finite cyclic-shift ambiguity. In the continuous $\mathrm{SO}(2)$ setting, this terminal step is not needed in the same form: the anchor phase represents the continuous rotation degree of freedom and can be fixed by a gauge choice. Thus, after bandlimiting, the $\mathrm{SO}(2)$ model  admits a recovery closely analogous to the cyclic result: the first three moments generically determine the $\mathrm{O}(2)$-orbit of the unknown bandlimited signal. 

\subsection{Future work}

Several directions remain open. On the theoretical side, a natural next step is to understand how far the model studied in this work extends beyond the specific reflection-symmetric projection considered here. More generally, one would like to characterize classes of projections for which low-order moments of the projected signal preserve the invariant information of the original orbit-recovery problem~\cite{balanov2026group} or transform it into the invariant structure of an enlarged symmetry group. This question is particularly relevant for higher-dimensional projection-based inverse problems, including cryo-EM and related settings, where projection masks orientation information but low-order moments may still retain enough structure for recovery.

A second important direction is the design of statistically efficient reconstruction algorithms. The finite-sample procedures studied here are deliberately simple benchmarks: EM provides a likelihood-based reference, while direct moment and Fourier-cosine moment optimization test two natural moment representations. A useful next step would be to develop reconstruction methods that retain the constructive identifiability guarantees underlying Algorithm~\ref{alg:projected-mra-constructive-recovery}, but replace its exact algebraic pruning steps with statistically stable procedures that use the full noisy moment information.

\section*{Data Availability}
The detailed implementation and code are available at \href{https://github.com/AmnonBa/projected-MRA}{https://github.com/AmnonBa/projected-MRA}.

\section*{Acknowledgment}
T.B. and D.E. are supported in part by BSF under Grant 2020159. T.B. is also supported in part by NSF-BSF under Grant 2024791, and in part by ISF under Grant 1924/21. DE is also supported by NSF DMS2205626.

\bibliographystyle{plain}

\begin{appendices}

\end{appendices}

\end{document}